\documentclass[journal, twocolumn]{IEEEtran}

\ifCLASSINFOpdf
\else
\fi

\usepackage[cmex10]{amsmath}

\usepackage{amssymb,amsthm}
\usepackage{color}
\usepackage{graphics}

\usepackage{multirow}
\usepackage{hhline}

\usepackage{tikz}
\usetikzlibrary{angles,calc,intersections,quotes,arrows.meta}
\usetikzlibrary{calc, arrows,backgrounds,positioning,fit}
\usetikzlibrary{decorations.markings}
\usetikzlibrary{shapes.geometric}
\usetikzlibrary{decorations.pathreplacing}
\usepackage{pgfplots}
\tikzset{middlearrow/.style={
        decoration={markings,
            mark= at position 0.5 with {\arrow{#1}} ,
        },
        postaction={decorate}
    }
}

\usepackage{mathtools}

\newcommand{\verteq}{\rotatebox{90}{$\,=$}}
\newcommand{\equalto}[2]{\underset{\scriptstyle\overset{\mkern4mu\verteq}{#2}}{#1}}

\allowdisplaybreaks

\usepackage[noadjust]{cite}

\makeatletter
\newtheorem*{rep@theorem}{\rep@title}
\newcommand{\newreptheorem}[2]{%
\newenvironment{rep#1}[1]{%
 \def\rep@title{#2 \ref{##1}}%
 \begin{rep@theorem}}%
 {\end{rep@theorem}}}
\makeatother

\theoremstyle{definition}
\newtheorem{definition}{Definition}

\newtheorem{example}[definition]{Example}

\theoremstyle{plain}
\newtheorem{theorem}{Theorem}
\newreptheorem{theorem}{Theorem}
\newtheorem{proposition}[definition]{Proposition}
\newtheorem{lemma}[definition]{Lemma}
\newtheorem{remark}[definition]{Remark}
\newtheorem{corollary}[definition]{Corollary}

\begin{document}
\title{Sum-Rank BCH Codes and\\Cyclic-Skew-Cyclic Codes}
\author{Umberto~Mart{\'i}nez-Pe\~{n}as%,~\IEEEmembership{Member,~IEEE}%~\IEEEmembership{Member,~IEEE}% <-this % stops a space
%\thanks{This work is supported by ?? No one. }
\thanks{U. Mart{\'i}nez-Pe\~{n}as is with the Institute of Computer Science and the Institute of Mathematics, University of Neuch{\^a}tel, 2000 Neuch{\^a}tel, Switzerland. (e-mail: umberto.martinez@unine.ch).}}

%Copyright (c) 2014 IEEE. Personal use of this material is permitted. However, permission to use this material for any other purposes must be obtained from the IEEE by sending a request to pubs-permissions@ieee.org.}}% <-this % stops a space
% \thanks{Manuscript received April 19, 2005; revised September 17, 2014.}}


\markboth{}%
{}



\maketitle


%% I changed the abstract to passive voice...  FK

\begin{abstract}
In this work, cyclic-skew-cyclic codes and sum-rank BCH codes are introduced. Cyclic-skew-cyclic codes are characterized as left ideals of a suitable non-commutative finite ring, constructed using skew polynomials on top of polynomials (or vice versa). Single generators of such left ideals are found, and they are used to construct generator matrices of the corresponding codes. The notion of defining set is introduced, using pairs of roots of skew polynomials on top of poynomials. A lower bound (called sum-rank BCH bound) on the minimum sum-rank distance is given for cyclic-skew-cyclic codes whose defining set contains certain consecutive pairs. Sum-rank BCH codes, with prescribed minimum sum-rank distance, are then defined as the largest cyclic-skew-cyclic codes whose defining set contains such consecutive pairs. The defining set of a sum-rank BCH code is described, and a lower bound on its dimension is obtained. Thanks to it, tables are provided showing that sum-rank BCH codes beat previously known codes for the sum-rank metric for binary $ 2 \times 2 $ matrices (i.e., codes whose codewords are lists of $ 2 \times 2 $ binary matrices, for a wide range of list lengths that correspond to the code length). Finally, a decoder for sum-rank BCH codes up to half their prescribed distance is obtained.
\end{abstract}


% \fi
\IEEEpeerreviewmaketitle

\section{Introduction} \label{sec intro}

\IEEEPARstart{C}{odes} correcting errors and/or erasures with respect to the sum-rank metric have found applications in universal error correction and security in multishot network coding \cite{multishot, secure-multishot}, rate-diversity optimal space-time codes \cite{space-time-kumar, mohannad}, partial-MDS codes for repair in distributed storage \cite{universal-lrc}, compartmented secret sharing \cite{compartment}, and private information retrieval on partial-MDS-coded databases \cite{pirlrc}. They may also find applications in extending McEliece's public-key cryptosystem \cite{mceliece-crypto}, or in a multishot or multilayer version of crisscross error and erasure correction, extending \cite{roth}. The sum-rank metric is also of theoretical interest as it recovers the Hamming metric \cite{hamming} and the rank metric \cite{delsartebilinear, gabidulin} as particular cases (see Subsection \ref{subsec sum-rank metric}).

Similarly to other metrics, the objective is to obtain codes with large size, large minimum sum-rank distance and small finite-field size, together with an efficient error-correcting algorithm (unless we are only concerned with erasures). As we show next, one of the main difficulties in applying the known code constructions in practice is that they are defined over rather large finite fields. 

A few constructions of convolutional codes for the sum-rank metric are known \cite{wachter, wachter-convolutional, mahmood-convolutional, mrd-convolutional, mrd-convolutional2}, but they are tailored mainly for erasure correction under erasure patterns common in streaming scenarios. Other constructions include concatenations of Hamming-metric convolutional codes with rank-metric block codes \cite{napp-sidorenko}, and multi-level constructions over rank-metric block codes \cite{multishot}. However, for such constructions, the number $ m $ of rows per matrix in the definition of the sum-rank metric (see Subsection \ref{subsec sum-rank metric}) needs to be large in order to use meaningful rank-metric component codes. Moreover, the parameter $ m $ appears as an exponent in the finite-field sizes of such constructions (convolutional code constructions \cite{wachter, wachter-convolutional, mahmood-convolutional, mrd-convolutional, mrd-convolutional2} suffer also from large finite-field size exponents). Therefore, finite fields of size $ 2^2 = 4 $ ($ m = 2 $), for instance, are not attainable by such techniques for non-trivial code parameters. Note that small values for $ m $, such as $ m = 2 $ or $ 3 $, are also of interest in the applications (they correspond to $ m $ outgoing links in multishot linearly coded networks \cite{multishot, secure-multishot}, and to locality $ m $ in locally repairable codes \cite{universal-lrc}).

The first known block codes whose minimum sum-rank distance attains the Singleton bound and have sub-exponential field sizes are linearized Reed-Solomon codes \cite{linearizedRS, caruso}. Such codes recover as particular cases (generalized) Reed-Solomon codes \cite{reed-solomon} and Gabidulin codes \cite{gabidulin, roth} whenever the sum-rank metric recovers the Hamming metric and the rank metric, respectively. Linearized Reed-Solomon codes may be defined over finite fields of size $ \Theta( \ell^m) $, where $ m $ is as above and $ \ell $ is the number of terms in the sum defining the sum-rank metric (see Subsection \ref{subsec sum-rank metric}) or, equivalently, the list size if we see codewords as lists of matrices (\ref{eq matrix representation codewords}). Furthermore, $ m $ can be arbitrary small, e.g., $ m = 1 $, $ 2 $ or $ 3 $, while the code length $ n = \ell m $, dimension $ k $ and minimum sum-rank distance $ d = n - k + 1 $ may be arbitrarily large. However, $ \ell $ still grows linearly as either $ n $, $ k $ or $ d $ grows, and small finite-field sizes such as $ 2^2 = 4 $ are not attainable either by such codes.

To tackle this issue, subfield subcodes of linearized Reed-Solomon codes were considered in \cite[Sec. VII]{universal-lrc}. Their minimum sum-rank distance, over the smaller finite-field extension, is at least the minimum sum-rank distance, over the larger finite-field extension, of the linearized Reed-Solomon codes (see Subsection \ref{subsec sum-rank metric} and Lemma \ref{lemma min sum-rank distance subextension subcodes}). However, the only estimate on their dimension (i.e., code size) considered in \cite{universal-lrc} is the well-known Delsarte's lower bound \cite{delsarte} (see (\ref{eq lower bound delsarte})). Note that the sum-rank metric for subfield subcodes considered in \cite{universal-lrc} is not the corresponding extension of the rank metric defined over subfield subcodes considered in \cite{loidreau-subfield}.

In this work, we describe two new families of linear codes: \textit{Cyclic-skew-cyclic codes} (Definition \ref{def CSC codes}) and one of its subfamilies, \textit{sum-rank BCH codes} (Definition \ref{def SR-BCH codes}). The latter codes have a prescribed minimum sum-rank distance thanks to being subfield subcodes of certain linearized Reed-Solomon codes. They can be defined over arbitrarily small finite fields, including $ 2^2 = 4 $ for $ m = 2 $ and for arbitrarily large $ \ell $, with code length $ n = \ell m $ (for such parameters, codewords can be seen as lists of $ 2 \times 2 $ binary matrices with list length $ \ell $, see (\ref{eq matrix representation codewords}) and Appendix \ref{app tables}). In addition, we show that their dimension is in most cases higher than that obtained by Delsarte's lower bound (see Appendix \ref{app tables}) for a given prescribed minimum sum-rank distance. To the best of our knowledge, the code dimensions that we obtain in Appendix \ref{app tables} are the highest known so far for $ m = 2 $ and the finite-field size $ 2^2 = 4 $ (i.e., codewords are lists of $ 2 \times 2 $ binary matrices), for the given code lengths and minimum sum-rank distances.

\subsection{Similar Codes in the Literature} \label{subsec similar codes}

Here, we discuss related codes beyond other known codes endowed with the sum-rank metric \cite{wachter, wachter-convolutional, mahmood-convolutional, mrd-convolutional, mrd-convolutional2, napp-sidorenko, multishot}, which were discussed above. 

Cyclic-skew-cyclic codes form a family of codes that recover as particular cases classical cyclic codes \cite[Sec. 7.2]{macwilliamsbook} \cite[Ch. 4]{pless} and skew-cyclic codes \cite[Page 6]{gabidulin} \cite[Def. 1]{skewcyclic1}, whenever the sum-rank metric recovers the Hamming metric and the rank metric, respectively (note that skew-cyclic codes are also good for the Hamming metric \cite{skewcyclic1}). Mathematically, cyclic-skew-cyclic codes can be characterized simultaneously as classical cyclic codes over some non-commutative finite ring, and as skew-cyclic codes over some commutative finite ring (Theorem \ref{th algebraic char}), using polynomials on top of skew polynomials (\ref{eq non-comm ring first}) or vice versa (\ref{eq non-comm ring second}). 

The study of cyclic codes over general finite rings was initiated in \cite{blake}. Not many works exist treating cyclic codes over non-commutative finite rings \cite{greferath-cyclic, hirano, greferath-rings}, and they deal with general properties rather than particular code constructions. 

A wide range of works deal with skew-cyclic codes over commutative finite rings that are not fields \cite{skewcyclic4, skewcyclic5, skewcyclic6, skewcyclic7}. However, no work has yet treated skew-cyclic codes over the commutative finite ring $ \mathbb{F}[x] / (x^\ell - 1) $, to the best of our knowledge.

Cyclic convolutional codes \cite{piret, heide-schmale} are similar to cyclic-skew-cyclic codes in that they also make use of skew polynomials over the commutative finite ring $ \mathbb{F}[x] / (x^\ell - 1) $, for a finite field $ \mathbb{F} $. However, they yield (infinite) convolutional codes rather than (finite) block codes.

Sum-rank BCH codes are both cyclic-skew-cyclic codes and subfield subcodes of certain linearized Reed-Solomon codes. Furthermore, they recover classical BCH codes \cite{bose, hocquenghem} and rank-metric BCH codes \cite{skewcyclic3}, whenever the sum-rank metric recovers the Hamming metric and the rank metric, respectively.

\subsection{Main Results} \label{subsec main results}

We introduce cyclic-skew-cyclic codes in Definition \ref{def CSC codes}, and characterize them as left ideals over a suitable non-commutative finite ring in Theorem \ref{th algebraic char}. 

In Theorem \ref{th single generator}, we show that such ideals are generated by a unique element satisfying certain properties, called minimal generator skew polynomial (Definition \ref{def minimal gen and check skew pol}). We show how to use such a left-ideal generator to obtain a generator matrix of the corresponding code in Theorem \ref{th generator matrix}. 

We then introduce the notion of defining set (Definition \ref{def defining set}) and show that they characterize the corresponding cyclic-skew-cyclic code (Theorem \ref{th defining set}). In Theorem \ref{th dimensions from defining set}, we show how to obtain the dimension of a cyclic-skew-cyclic code from its defining set. 

In Theorem \ref{th CSC linearized RS codes}, we obtain a family of cyclic-skew-cyclic linearized Reed-Solomon codes. Using such codes, we obtain in Theorem \ref{th SR-BCH bound} a lower bound (called \textit{sum-rank BCH bound}) on the minimum sum-rank distance of cyclic-skew-cyclic codes whose defining set contains consecutive pairs. Sum-rank BCH codes are then defined as the largest codes satisfying such a property (Definition \ref{def SR-BCH codes}). 

In Theorem \ref{th SR-BCH structure}, we describe the defining set of a sum-rank BCH code based on the pairs in its definition. Using this result and Theorem \ref{th dimensions from defining set}, we obtain in Theorem \ref{th lower bound} a lower bound on the dimension of a sum-rank BCH code that can be easily computed from the pairs in its definition. 

In Subsection \ref{subsec SR-BCH decoding} we show that sum-rank BCH codes can be decoded with respect to the sum-rank metric up to half their prescribed distance by decoding the larger linearized Reed-Solomon code.

In Appendix \ref{app tables}, we obtain tables with lower and upper bounds on the parameters of sum-rank BCH codes for $ m = 2 $ and the finite-field size $ 2^2 = 4 $. These tables show that the introduced sum-rank BCH codes beat previously known code dimensions for a given prescribed distance, for a wide range of code lengths.

\subsection{Organization} \label{subsec organization}

The remainder of the manuscript is organized as follows. In Section \ref{sec preliminaries}, we provide some preliminaries, mainly the definitions of the sum-rank metric and skew polynomial rings. In Section \ref{sec CSC codes}, we introduced cyclic-skew-cyclic codes and characterize them as left ideals. In Section \ref{sec generators}, we find a single generator for such left ideals and the corresponding generator matrix. In Section \ref{sec roots}, we introduce defining sets of cyclic-skew-cyclic codes, after defining appropriate evaluation maps for skew polynomials. We conclude the section by computing the dimension of a cyclic-skew-cyclic code from its defining set. In Section \ref{sec CSC LRS codes}, we revisit linearized Reed-Solomon codes and find a subfamily of such codes formed by cyclic-skew-cyclic codes. In Section \ref{sec SR BCH codes}, we provide the sum-rank BCH bound and the definition of sum-rank BCH codes. We then describe their defining sets and conclude by giving a lower bound on their dimensions. The section concludes by showing how to decode sum-rank BCH codes up to half their prescribed distance. Section \ref{sec conclusion} concludes the manuscript.

\section{Preliminaries and General Setting} \label{sec preliminaries}

\subsection{Main Finite Fields in This Work} \label{subsec finite fields this work}

In this work, $ q_0 $ denotes the power of a prime number $ p $. We also fix positive integers $ m $, $ \ell $ and $ s $, and denote $ q = q_0^s $. We will assume that the finite field with $ q $ elements
$$ \mathbb{F}_q = \mathbb{F}_{q_0^s} $$
contains all $ \ell $th \textit{roots of unity}, i.e., all roots of the polynomial $ x^\ell - 1 \in \mathbb{F}_p[x] $. This holds if, and only if, $ q-1 $ is divisible by $ \ell $, see \cite[p. 110]{pless}. Observe that, if $ \ell = 1 $, then we may assume that $ s = 1 $. We refer to \cite{lidl} for generalities on finite fields. 

Throughout the manuscript, we will consider \textit{sum-rank metrics} (see Definition \ref{def sum-rank metric} below), for which we will consider the two finite-field extensions  
$$ \mathbb{F}_{q_0} \subseteq \mathbb{F}_{q_0^m} \quad \textrm{and} \quad \mathbb{F}_q \subseteq \mathbb{F}_{q^m}. $$
To relate these four fields, the following directed graph might be helpful, where $ K \longrightarrow L $ means that the field $ L $ is a field extension of the field $ K $, i.e., $ K $ is a subfield of $ L $:

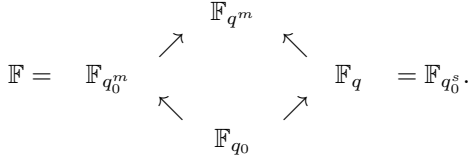
\begin{displaymath}
\begin{array}{rcccl}
 & & \mathbb{F}_{q^m} & & \\
  & \nearrow & & \nwarrow & \\
\mathbb{F} = \quad \mathbb{F}_{q_0^m} & & & & \mathbb{F}_{q} \quad = \mathbb{F}_{q_0^s} . \\
 & \nwarrow & & \nearrow & \\
  & & \mathbb{F}_{q_0} & & \\
\end{array}
\end{displaymath}

Since we will define our main codes over $ \mathbb{F}_{q_0^m} $, we will simply denote
\begin{equation} \label{eq main finite field}
\mathbb{F} = \mathbb{F}_{q_0^m}.
\end{equation}

\subsection{The Sum-Rank Metric} \label{subsec sum-rank metric}

We fix another positive integer $ N $ and define a \textit{code length partition} 
\begin{equation} \label{eq sum-rank length partition}
n = \ell N  = \underbrace{N + N + \cdots + N}_{\ell \textrm{ times}} .
\end{equation}
The integer $ n $ will be the length of our main codes, and (\ref{eq sum-rank length partition}) will be the \textit{length partition} that we will use in order to define the sum-rank metric, which was defined in \cite[Sec. III-D]{multishot} under the name \textit{extended rank distance}.

\begin{definition} [\textbf{Sum-Rank Metric \cite{multishot}}] \label{def sum-rank metric}
Let $ K \subseteq L $ be a field extension (in our case, $ \mathbb{F}_{q_0} \subseteq \mathbb{F}_{q_0^m} $ or $ \mathbb{F}_q \subseteq \mathbb{F}_{q^m} $). Consider vectors $ \mathbf{c} \in L^n $ to be subdivided according to the partition $ n = \ell N $ as in (\ref{eq sum-rank length partition}):
\begin{equation*} 
\begin{array}{rll}
\mathbf{c} & = (\mathbf{c}^{(0)}, \mathbf{c}^{(1)}, \ldots, \mathbf{c}^{(\ell - 1)}) & \in L^n, \textrm{ where} \\
\mathbf{c}^{(i)} & = (c^{(i)}_0, c^{(i)}_1, \ldots, c^{(i)}_{N - 1}) & \in L^N, 
\end{array}
\end{equation*} 
for $ i = 0,1, \ldots, \ell - 1 $. We define the \textit{sum-rank weight} for the extension $ K \subseteq L $ (i.e., the pair $ (K, L) $) and length partition $ n = \ell N $ as in (\ref{eq sum-rank length partition})  (i.e., the pair $ (\ell, N) $) as the map $ {\rm wt}_{SR} : L^n \longrightarrow \mathbb{N} $ given by
$$ {\rm wt}_{SR}(\mathbf{c}) = \sum_{i = 0}^{\ell - 1} \dim_K \left( \left\langle c^{(i)}_0, c^{(i)}_1, \ldots, c^{(i)}_{N-1} \right\rangle_K \right), $$
for all $ \mathbf{c} \in L^n $ subdivided as above. Here, $ \dim_K( \cdot ) $ denotes dimension over $ K $ and $ \left\langle \cdot \right\rangle_K $ denotes $ K $-linear span. We define the \textit{sum-rank distance} for the same field extension and length partition as the map $ {\rm d}_{SR} : L^n \times L^n \longrightarrow \mathbb{N} $ given by
$$ {\rm d}_{SR} (\mathbf{c}, \mathbf{d}) = {\rm wt}_{SR}(\mathbf{c} - \mathbf{d}), $$
for all $ \mathbf{c}, \mathbf{d} \in L^n $. Finally, for a code $ \mathcal{C} \subseteq L^n $ (a \textit{code} will simply be any subset of $ L^n $, linear or not), we define its minimum sum-rank distance as
$$ {\rm d}_{SR}(\mathcal{C}) = \min \{ {\rm d}_{SR}(\mathbf{c}, \mathbf{d}) \mid \mathbf{c}, \mathbf{d} \in \mathcal{C}, \mathbf{c} \neq \mathbf{d} \}. $$
\end{definition}

To explain the name \textit{sum-rank metric}, observe that, if $ m = \dim_K(L) < \infty $, then 
$$ {\rm wt}_{SR}(\mathbf{c}) = \sum_{i = 0}^{\ell - 1} {\rm Rk}_K \left( \begin{array}{cccc}
c_{0,0}^{(i)} & c_{0,1}^{(i)} & \ldots & c_{0,N-1}^{(i)} \\
c_{1,0}^{(i)} & c_{1,1}^{(i)} & \ldots & c_{1,N-1}^{(i)} \\
\vdots & \vdots & \ddots & \vdots \\
c_{m-1,0}^{(i)} & c_{m-1,1}^{(i)} & \ldots & c_{m-1,N-1}^{(i)} \\
\end{array} \right), $$
where $ {\rm Rk}_K(\cdot) $ denotes rank over $ K $, $ \mathbf{c} \in L^n $ is subdivided as in Definition \ref{def sum-rank metric}, and
$$ c^{(i)}_j = \sum_{h = 0}^{m - 1} c^{(i)}_{h, j} \beta_h, $$
where $ c^{(i)}_{h, j} \in K $, for $ h = 0,1,\ldots, m-1 $, for $ j = 0,1, \ldots, N-1 $ and for $ i = 0,1, \ldots, \ell-1 $, and where $ \{ \beta_0, \beta_1, \ldots, \beta_{m - 1} \} $ is an ordered basis of $ L $ over $ K $. With such a matrix representation, we may consider $ \mathbf{c} \in L^n $ as a list of $ \ell $ matrices
\begin{equation}
(C_1, C_2, \ldots, C_\ell) \in \left( K^{m \times N} \right)^\ell . 
\label{eq matrix representation codewords}
\end{equation}
Considering codes as subsets of $ \left( K^{m \times N} \right)^\ell $, endowed with the sum-rank metric as above, may be important in some applications, but becomes cumbersome in our study.

The sum-rank metric recovers the Hamming metric if $ m = N = 1 $ \cite[Ex. 36]{linearizedRS} and the rank metric if $ \ell = 1 $ \cite[Ex. 37]{linearizedRS}. See also \cite[Sec. II-A]{universal-lrc}. We will show throughout the paper how our results particularize to these two important cases.

In most applications of the sum-rank metric, the objective is to obtain codes having simultaneously a large size and a large minimum sum-rank distance, for a given code alphabet (pair $ (K,L) $) and length (pair $ (\ell, N) $). As in the classical Hamming-metric case \cite[Th. 2.4.1]{pless}, there is a trade-off between the size and minimum sum-rank distance of any code (linear or not) given by the Singleton bound \cite[Cor. 2]{universal-lrc}. We refer to \cite{macwilliamsbook, pless} for generalities on codes.

\begin{proposition} [\textbf{\cite{universal-lrc}}]
With notation as in Definition \ref{def sum-rank metric}, for a finite-field extension $ K \subseteq L $, and for a code $ \mathcal{C} \subseteq L^n $ (which may be linear or not), it holds that
\begin{equation}
|\mathcal{C}| \leq |L|^{n - {\rm d}_{SR}(\mathcal{C}) + 1}.
\label{eq singleton bound}
\end{equation}
\end{proposition}

Codes achieving equality in (\ref{eq singleton bound}) are called \textit{maximum sum-rank distance} (MSRD) codes. Maximum rank distance (MRD) codes (e.g., Gabidulin codes \cite{gabidulin, roth}) are also MSRD codes, but any MRD code requires $ |L| \geq 2^{\ell N} $, exponential in the parameters $ n $, $ \ell $ and $ N $. The first MSRD codes for sub-exponential field sizes $ |L| $ in $ n $ and $ \ell $ are linearized Reed-Solomon codes \cite{linearizedRS, caruso}. See Subsection \ref{subsec revisiting lin RS codes}. They require $ |L| = \Theta (\ell^N) $. By considering subfield subcodes, one can reduce the base $ \ell $ at the expense of possibly reducing code size (relative to field size) and minimum sum-rank distance \cite[Sec. VII]{universal-lrc}. The objective of this manuscript is to study the structure of one family of subfield subcodes of linearized Reed-Solomon codes (Section \ref{sec SR BCH codes}). In particular, we will obtain a better estimate on their dimensions than previously known \cite[Cor. 9]{universal-lrc}. This will lead to codes with larger size and minimum sum-rank distance than previous codes with $ |L| \ll \ell^N $. 

To consider subfield subcodes and the sum-rank metric, we will consider the two finite-field extensions $ \mathbb{F}_{q_0} \subseteq \mathbb{F}_{q_0^m} $ and $ \mathbb{F}_q \subseteq \mathbb{F}_{q^m} $, as in Subsection \ref{subsec finite fields this work}. Since we will only consider the length partition $ n = \ell N $ as in (\ref{eq sum-rank length partition}), for $ \ell $ and $ N $ fixed, we do not need to write the dependency of the sum-rank metric on the field extension and length partition. We will simply denote by
\begin{equation}
\begin{split}
{\rm wt}_{SR} & : \mathbb{F}_{q^m}^n \longrightarrow \mathbb{N}, \textrm{ and} \\
{\rm wt}_{SR}^0 & : \mathbb{F}_{q_0^m}^n \longrightarrow \mathbb{N},
\end{split}
\label{eq sum-rank weights in this work}
\end{equation}
the sum-rank weights for the extensions $ \mathbb{F}_q \subseteq \mathbb{F}_{q^m} $ and $ \mathbb{F}_{q_0} \subseteq \mathbb{F}_{q_0^m} $, respectively (analogously for the corresponding metrics).

\subsection{Skew Polynomial Rings} \label{subsec skew pols}

Skew polynomial rings over division rings (i.e., commutative or non-commutative fields) were originally defined by Ore \cite{ore}. However, in this work, we will need to consider skew polynomial rings over commutative rings which are not necessarily integral domains. We refer to \cite{lang-under} for basic Algebra and to \cite{lam-book} for non-commutative Algebra.

Let $ R $ be an arbitrary commutative ring (with identity) and let $ \sigma : R \longrightarrow R $ be a ring automorphism (we always assume that $ \sigma(1) = 1 $). The skew polynomial ring $ R[z; \sigma] $ is defined as the free module over $ R $ with basis $ \{ z^i \mid i \in \mathbb{N} \} $, being $ \mathbb{N} = \{ 0 ,1,2, \ldots \} $, and with product given by the rule $ z^i z^j = z^{i+j} $, for all $ i,j \in \mathbb{N} $, and the rule
\begin{equation} \label{eq non-comm rule for skew pols}
z a = \sigma(a) z,
\end{equation}
for all $ a \in R $. Then $ R[z; \sigma] $ is a ring with identity $ 1 = z^0 $. It is commutative if, and only if, $ \sigma = {\rm Id} $ (the identity automorphism), and it is an integral domain if, and only if, so is $ R $. Furthermore, if $ R $ is a field, then $ R[z; \sigma] $ is both a left and right Euclidean domain.

To define skew polynomial rings, we will mostly consider the field automorphism
\begin{equation}
\begin{array}{rcl}
\sigma : \mathbb{F}_{q^m} & \longrightarrow & \mathbb{F}_{q^m} \\
a & \mapsto & a^q,
\end{array}
\end{equation}
for all $ a \in \mathbb{F}_{q^m} $. We will also consider $ \sigma $ restricted to some subfields of $ \mathbb{F}_{q^m} $, and we will also extend $ \sigma $ to polynomial rings over such subfields. Regardless of its domain, we will always use the letter $ \sigma $, as its definition can be easily inferred from the context.

\section{Cyclic-Skew-Cyclic Codes} \label{sec CSC codes}

\subsection{The Definition} \label{subsec definitions}

Recall that $ \mathbb{F} = \mathbb{F}_{q_0^m} $, as defined in (\ref{eq main finite field}), and that we consider vectors in $ \mathbb{F}^n $ to be subdivided as follows:
\begin{equation} \label{eq subdivision vectors}
\begin{array}{rll}
\mathbf{c} & = (\mathbf{c}^{(0)}, \mathbf{c}^{(1)}, \ldots, \mathbf{c}^{(\ell - 1)}) & \in \mathbb{F}^n, \textrm{ where} \\
\mathbf{c}^{(i)} & = (c^{(i)}_0, c^{(i)}_1, \ldots, c^{(i)}_{N-1}) & \in \mathbb{F}^N,
\end{array}
\end{equation}
for $ i = 0,1, \ldots, \ell - 1 $. With this notation, we may define the following operators.

\begin{definition} [\textbf{Shifting Operators}] \label{def CSC operators}
The \textit{cyclic inter-block shifting operator} $ \varphi : \mathbb{F}^n \longrightarrow \mathbb{F}^n $ is defined as
$$ \varphi \left( \mathbf{c}^{(0)}, \mathbf{c}^{(1)}, \ldots, \mathbf{c}^{(\ell - 1)} \right) = \left( \mathbf{c}^{(\ell - 1)}, \mathbf{c}^{(0)}, \ldots, \mathbf{c}^{(\ell-2)} \right) . $$
The \textit{skew-cyclic intra-block shifting operator} $ \phi : \mathbb{F}^n \longrightarrow \mathbb{F}^n $ is defined as
$$ \phi \left( \mathbf{c}^{(0)}, \mathbf{c}^{(1)}, \ldots, \mathbf{c}^{(\ell - 1)} \right) = $$
$$ \left( \psi(\mathbf{c}^{(0)}), \psi(\mathbf{c}^{(1)}), \ldots, \psi(\mathbf{c}^{(\ell - 1)}) \right), $$
where $ \psi : \mathbb{F}^N \longrightarrow \mathbb{F}^N $ is the classical skew-cyclic shifting operator, given by
$$ \psi \left( c_0, c_1, \ldots, c_{N - 1} \right) = \left( \sigma(c_{N- 1}), \sigma(c_0), \ldots, \sigma(c_{N-2}) \right) . $$
\end{definition}

The operators in Definition \ref{def CSC operators} can be trivially extended to any field $ L $ and endomorphism $ \sigma : L \longrightarrow L $. These operators depend on the length partition $ n = \ell N $, the field $ L $ and the field endomorphism $ \sigma $. However, we will not write this dependency for simplicity. 

The classical cyclic shifting operator \cite[Sec. 7.2]{macwilliamsbook} \cite[Ch. 4]{pless} is recovered from $ \varphi $ by setting $ N = 1 $, whereas the classical skew-cyclic shifting operator \cite[Page 6]{gabidulin} \cite[Def. 1]{skewcyclic1} is recovered from $ \phi $ by setting $ \ell = 1 $. Moreover, $ \phi $ and $ \varphi $ become the identity maps if $ m = N = 1 $ and $ \ell = 1 $, respectively.

We may now define cyclic-skew-cyclic codes.

\begin{definition} [\textbf{Cyclic-Skew-Cyclic Codes}] \label{def CSC codes}
We say that a code $ \mathcal{C} \subseteq \mathbb{F}^n $ is a \textit{cyclic-skew-cyclic code}, or a \textit{CSC code} for short, if it is a linear code (linear over $ \mathbb{F} $), and
$$ \varphi(\mathcal{C}) \subseteq \mathcal{C} \quad \textrm{and} \quad \phi(\mathcal{C}) \subseteq \mathcal{C}. $$
Analogously for any field $ L $, instead of $ \mathbb{F} $, and any field endomorphism $ \sigma : L \longrightarrow L $.
\end{definition}

By the observations above, classical cyclic codes \cite[Sec. 7.2]{macwilliamsbook} \cite[Ch. 4]{pless} and classical skew-cyclic codes \cite[Page 6]{gabidulin} \cite[Def. 1]{skewcyclic1} are recovered from CSC codes by setting $ m = N = 1 $ and $ \ell = 1 $, respectively.

\subsection{Algebraic Characterizations} \label{subsec algebraic char}

We now give an algebraic description of CSC codes, which in turn recovers those of classical cyclic codes and skew-cyclic codes. First, we define 
\begin{equation} \label{eq non-comm ring first}
\mathcal{R}^\prime = \frac{\mathcal{S}^\prime[x]}{(x^\ell - 1)}, \quad \textrm{where} \quad \mathcal{S}^\prime = \frac{\mathbb{F}[z; \sigma]}{(z^N - 1)}.
\end{equation}
The use of the prime symbol is due to the fact that we first consider $ \mathcal{R}^\prime $, since it is simpler to describe and corresponds to the length partition in (\ref{eq sum-rank length partition}). However, we will use an alternative ring throughout the manuscript. We will assume that $ \mathcal{S}^\prime $ is a ring, which holds if, e.g., $ m $ divides $ N $. We are considering the polynomial ring $ \mathcal{S}^\prime[x] $ as usual, where the variable $ x $ commutes with all elements. In particular, we have that
$$ zx = xz. $$
Ideals in $ \mathcal{S}^\prime[x] $ are assumed to be left ideals. Hence $ (x^\ell - 1) $ denotes the left ideal of $ \mathcal{S}^\prime[x] $ generated by $ x^\ell - 1 \in \mathcal{S}^\prime[x] $, even though this ideal in particular is two-sided (since $ xz = zx $ and $ xa = ax $ for all $ a \in \mathbb{F} $).

We may identify $ \mathbb{F}^n $ with $ \mathcal{R}^\prime $, as vector spaces over $ \mathbb{F} $, via the vector space isomorphism $ \mu : \mathbb{F}^n \longrightarrow \mathcal{R}^\prime $ given by
\begin{equation} \label{eq coordinate ismorphism first}
\mu \left( \mathbf{c}^{(0)}, \mathbf{c}^{(1)}, \ldots, \mathbf{c}^{(\ell - 1)} \right) = \left( \sum_{i=0}^{\ell - 1} \left( \sum_{j = 0}^{N-1} c^{(i)}_j z^j \right) x^i \right), 
\end{equation}
where $ \mathbf{c}^{(i)} = (c^{(i)}_0, c^{(i)}_1, \ldots, c^{(i)}_{N-1}) \in \mathbb{F}^N $, for $ i = 0,1, \ldots, \ell -1 $. We will often denote
$$ c(x, z) = \mu(\mathbf{c}). $$

It can be shown that $ \mathcal{R}^\prime $ is naturally isomorphic to
\begin{equation} \label{eq non-comm ring second}
\mathcal{R} = \frac{\mathcal{S}[z; \sigma]}{(z^N - 1)}, \quad \textrm{where} \quad \mathcal{S} = \frac{\mathbb{F}[x]}{(x^\ell - 1)},
\end{equation}
where $ \sigma : \mathbb{F} \longrightarrow \mathbb{F} $ is extended uniquely to $ \sigma : \mathcal{S} \longrightarrow \mathcal{S} $ by setting $ \sigma (x) = x $. 

This is the first time that we need to extend $ \sigma $ to a larger ring. First note that $ \sigma $ can be uniquely extended to $ \sigma : \mathbb{F}[x] \longrightarrow \mathbb{F}[x] $, satisfying that $ \sigma (x) = x $. Second, $ \sigma : \mathbb{F}[x] \longrightarrow \mathbb{F}[x] $ can be uniquely extended to 
\begin{equation} \label{eq extending theta to superring}
\sigma : \frac{\mathbb{F}[x]}{(m(x))} \longrightarrow \frac{\mathbb{F}[x]}{(m(x))},
\end{equation}
satisfying that $ \sigma(x) = x $, where $ m(x) \in \mathbb{F}[x] $, if and only if, 
$$ \sigma(m(x)) \in (m(x)). $$
This is clearly the case for $ m(x) = x^\ell - 1 $ since $ \sigma(x^\ell - 1) = x^\ell -  1 $. We will see less trivial cases later on.

Observe that $ \mathcal{S} $ is a commutative ring, but it is not an integral domain in general, thus $ \mathcal{R} $ is not an integral domain in general. Note that, in this case, the fact that $ zx = xz $ comes from rule (\ref{eq non-comm rule for skew pols}) in the definition of skew polynomial rings and $ \sigma (x) = x $.

We may also identify $ \mathbb{F}^n $ with $ \mathcal{R} $ as vector spaces over $ \mathbb{F} $. In this case, we may use $ \nu : \mathbb{F}^n \longrightarrow \mathcal{R} $ given by
\begin{equation} \label{eq coordinate ismorphism second}
\nu \left( \mathbf{c}^{(0)}, \mathbf{c}^{(1)}, \ldots, \mathbf{c}^{(\ell - 1)} \right) = \left( \sum_{j = 0}^{N-1} \left( \sum_{i=0}^{\ell - 1} c^{(i)}_j x^i \right) z^j \right), 
\end{equation}
where $ \mathbf{c}^{(i)} = (c^{(i)}_0, c^{(i)}_1, \ldots, c^{(i)}_{N-1}) \in \mathbb{F}^N $, for $ i = 0,1, \ldots, \ell - 1 $.

We have the following algebraic characterization of CSC codes.

\begin{theorem} \label{th algebraic char}
Let $ \mathcal{C} \subseteq \mathbb{F}^n $ be a code, which we do not assume to be linear. The following are equivalent:
\begin{enumerate}
\item
$ \mathcal{C} \subseteq \mathbb{F}^n $ is a CSC code.
\item
$ \mu(\mathcal{C}) \subseteq \mathcal{R}^\prime $ is a left ideal of $ \mathcal{R}^\prime $.
\item
$ \nu(\mathcal{C}) \subseteq \mathcal{R} $ is a left ideal of $ \mathcal{R} $.
\end{enumerate}
\end{theorem}
\begin{proof}
For all $ \mathbf{c} \in \mathbb{F}^n $, it holds that
$$ x \mu(\mathbf{c}) = \mu (\varphi(\mathbf{c})) \quad \textrm{and} \quad z \mu (\mathbf{c}) = \mu (\phi (\mathbf{c})). $$
Hence $ \mathcal{C} $ is $ \mathbb{F} $-linear, $ \varphi (\mathcal{C}) \subseteq \mathcal{C} $ and $ \phi (\mathcal{C}) \subseteq \mathcal{C} $ if, and only if, $ \mu (\mathcal{C}) $ is $ \mathbb{F} $-linear, $ x \mu (\mathcal{C}) \subseteq \mu (\mathcal{C}) $ and $ z \mu (\mathcal{C}) \subseteq \mu (\mathcal{C}) $. The latter three conditions are equivalent to $ \mu (\mathcal{C}) $ being a left ideal of $ \mathcal{R}^\prime $. This proves the equivalence of items 1 and 2. The equivalence with item 3 is analogous.
\end{proof}

Observe that Theorem \ref{th algebraic char} recovers the classical algebraic characterizations of cyclic codes and skew-cyclic codes by setting $ m = N = 1 $ and $ \ell = 1 $, respectively. In general, Theorem \ref{th algebraic char} states that a CSC code is simultaneously a cyclic code over the non-commutative finite ring $ \mathcal{S}^\prime $, and a skew-cyclic code over the commutative finite ring $ \mathcal{S} $.

Unless otherwise stated, we will identify $ \mathcal{C} $, $ \mu(\mathcal{C}) $  and $ \nu(\mathcal{C}) $, and we will denote by $ \mathcal{C} $ both the CSC code in $ \mathbb{F}^n $ and the left ideal of $ \mathcal{R}^\prime $ or $ \mathcal{R} $, indistinctly.

\section{Generators of Cyclic-Skew-Cyclic Codes} \label{sec generators}

\subsection{Finding a Single Generator as a Left Ideal} \label{subsec finding single gen pol}

Cyclic codes over an arbitrary (commutative or not) finite ring are not always principal ideals \cite{greferath-cyclic}, i.e., they do not always have a single generator. In fact, being \textit{principal} is equivalent to being \textit{splitting} \cite[Th. 3.2]{greferath-cyclic}. A similar result holds for \textit{cyclic convolutional codes} \cite[Th. 3.5]{heide-cyclic} \cite[Sec. 4]{heide-schmale}, which make use of skew polynomials but are not CSC codes. 

In this subsection, we will show that $ \mathcal{R} $ (thus $ \mathcal{R}^\prime $) is a \textit{principal left ideal ring} (PIR), assuming that $ p $ does not divide $ \ell $ (i.e., $ x^\ell - 1 $ has simple roots in $ \mathbb{F}_q $). To that end, we will find a minimal generator skew polynomial of a CSC code that will enable us to obtain a generator matrix (Subsection \ref{subsec consequences gen pol}) and obtain the defining set (Subsection \ref{subsec roots and generators}) of the CSC code.

From now on, let the polynomials $ m_1(x), $ $ m_2(x), $ $ \ldots, $ $ m_t(x) $ $ \in \mathbb{F}[x] $ form the unique irreducible decomposition
\begin{equation} \label{eq irreducible dec x^ell - 1}
x^\ell - 1 = m_1(x) m_2(x) \cdots m_t(x)
\end{equation}
of $ x^\ell - 1 $ in the polynomial ring $ \mathbb{F}[x] $. 

Recall from Subsection \ref{subsec finite fields this work} that we are assuming that $ \mathbb{F}_q = \mathbb{F}_{q_0^s} $ contains all roots of $ x^\ell - 1 $. Recall from Subsection \ref{subsec skew pols} that $ \sigma : \mathbb{F} \longrightarrow \mathbb{F} $ is given by $ \sigma(a) = a^q $, for $ a \in \mathbb{F} $. Combining these two facts, we deduce that $ \sigma(a) = a $, for every root $ a $ of $ x^\ell - 1 $. Since the roots of $ m_i(x) $ are a subset of the roots of $ x^\ell - 1 $, its coefficients lie in $ \mathbb{F}_q $ too, or in other words,
$$ \sigma(m_i(x)) = m_i(x), $$
for all $ i = 1,2, \ldots, t $. Hence, as in (\ref{eq extending theta to superring}), we may extend $ \sigma $ to 
$$ \sigma : \frac{\mathbb{F}[x]}{(m_i(x))} \longrightarrow \frac{\mathbb{F}[x]}{(m_i(x))}, $$
satisfying $ \sigma(x) = x $, for $ i = 1,2, \ldots, t $.

Assume from now on that $ x^\ell -1 $ has simple roots (i.e., $ p $ does not divide $ \ell $), then $ m_i(x) \neq m_j(x) $ if $ i \neq j $. By the B{\'e}zout identities in the polynomial ring $ (\mathbb{F} \cap \mathbb{F}_{q_0^s})[x] $ over the finite field $ \mathbb{F} \cap \mathbb{F}_{q_0^s} $, there exist $ a_i(x), b_i(x) \in (\mathbb{F} \cap \mathbb{F}_{q_0^s})[x] $ such that
\begin{equation} \label{eq bezout identities}
a_i(x) \left( \frac{x^\ell - 1}{m_i(x)} \right) + b_i(x) m_i(x) = 1,
\end{equation} 
for $ i = 1,2, \ldots, t $. Define the $ i $th \textit{primitive idempotent} \cite[Sec. 4.3]{pless} as
\begin{equation} \label{eq primitive idempotent}
\begin{split}
e_i(x) & = a_i(x) m_1(x) \cdots m_{i-1}(x) m_{i+1}(x) \cdots m_t(x) \\
& = a_i(x) \left( \frac{x^\ell - 1}{m_i(x)} \right) \in (\mathbb{F} \cap \mathbb{F}_{q_0^s})[x],
\end{split}
\end{equation}
for $ i = 1,2,\ldots, t $. Note that $ e_i(x) $ is not idempotent in $ \mathbb{F}[x] $, but its image in $ \mathbb{F}[x] / (x^\ell - 1) $ is (see (\ref{eq primitive idempotents are idempotents}) below). Note also that, since $ e_i(x) \in (\mathbb{F} \cap \mathbb{F}_{q_0^s})[x] $, then
$$ \sigma(e_i(x)) = e_i(x), $$
for $ i = 1,2, \ldots, t $, and in fact, it lies in the center of $ \mathcal{R} $ (or $ \mathcal{R}^\prime $) as it is constant in $ z $.

We may now state and prove the following result.

\begin{lemma} \label{lemma chinese remainder theorem}
Define the rings
$$ \mathcal{R}_i = \frac{\mathcal{S}_i [z; \sigma]}{(z^N - 1)}, \quad \textrm{where} \quad \mathcal{S}_i = \frac{\mathbb{F}[x]}{(m_i(x))}, $$
for $ i = 1,2, \ldots, t $. The natural maps
$$ \begin{array}{rrcl}
\rho : & \mathcal{R} & \longrightarrow & \mathcal{R}_1 \times \mathcal{R}_2 \times \cdots \times \mathcal{R}_t, \textrm{ and} \\
\tau : & \mathcal{S} [z; \sigma] & \longrightarrow & \left( \mathcal{S}_1 \times \mathcal{S}_2 \times \cdots \times \mathcal{S}_t \right) [z; \sigma] ,
\end{array} $$
given simply by projecting modulo the corresponding $ m_i(x) $, 
$$ \rho \left( \sum_{j=0}^{N-1} (f_j(x) + (x^\ell - 1)) z^j \right) = $$
$$ \sum_{j=0}^{N-1} \left( f_j(x) + (m_1(x)), \ldots, f_j(x) + (m_t(x)) \right) z^j, $$
and analogously for $ \tau $, are well-defined ring isomorphisms. In addition, it holds that
\begin{equation}
\rho(e_i(x)) = \tau(e_i(x)) = \mathbf{e}_i = (0, \ldots, 1, \ldots, 0) \in \mathbb{F}^t,
\label{eq primitive idempotents are idempotents}
\end{equation}
where $ \mathbf{e}_i \in \mathbb{F}^t $ has all of its components equal to $ 0 $ except its $ i $th component, which is equal to $ 1 $, for $ i = 1,2, \ldots, t $.

Finally, given $ f_i(x,z) \in \mathcal{R}_i $, for $ i = 1,2, \ldots, t $, the unique $ f(x,z) \in \mathcal{R} $, such that $ \rho(f(x,z)) = (f_1(x,z), $ $ f_2(x,z),$ $ \ldots, $ $ f_t(x,z)) $, is given by
$$ f(x,z) = \sum_{i=1}^t e_i(x) \widetilde{f}_i(x,z) = \sum_{i=1}^t \widetilde{f}_i(x,z) e_i(x), $$
where $ \widetilde{f}_i(x,z) \in \mathcal{R} $ is such that its projection onto $ \mathcal{R}_i $ is $ f_i(x,z) $. Analogously for $ \tau $.
\end{lemma}
\begin{proof}
By the \textit{Chinese Remainder Theorem} \cite[Sec. 10.9, Th. 5]{macwilliamsbook}, the natural map
\begin{equation} \label{eq chinese remainder theorem}
\equalto{\left( \frac{\mathbb{F}[x]}{(x^\ell - 1)} \right)}{\mathcal{S}}
 \longrightarrow \equalto{\left( \frac{\mathbb{F}[x]}{(m_1(x))} \right) \times \cdots \times \left( \frac{\mathbb{F}[x]}{(m_t(x))} \right)}{\mathcal{S}_1 \times \mathcal{S}_2 \times \cdots \times \mathcal{S}_t}
\end{equation}
is a ring isomorphism. The map $ \tau $ is just extending such a ring isomorphism to the corresponding skew polynomial rings. This can be done since $ \sigma $ commutes with the map in (\ref{eq chinese remainder theorem}), which holds since $ \sigma(x) = x $, for $ \sigma $ defined over all domains $ \mathcal{S}, \mathcal{S}_1, \mathcal{S}_2, \ldots, \mathcal{S}_t $.

Next, there is a natural surjective ring morphism 
$$ \left( \mathcal{S}_1 \times \mathcal{S}_2 \times \cdots \times \mathcal{S}_t \right) [z; \sigma] \longrightarrow \mathcal{R}_1 \times \mathcal{R}_2 \times \cdots \times \mathcal{R}_t. $$
By composing it with $ \tau $, we obtain the ring morphism
$$ \rho_0 : \left( \frac{\mathbb{F}[x]}{(x^\ell - 1)} \right) [z; \sigma] \longrightarrow \mathcal{R}_1 \times \mathcal{R}_2 \times \cdots \times \mathcal{R}_t. $$
Finally, we can extend $ \rho_0 $ to $ \rho $ by noting that $ \rho_0 (z^N - 1) = 0 $. This shows that $ \rho $ is well-defined and a ring morphism.

The fact that $ \rho $ and $ \tau $ are bijective, i.e., ring isomorphisms, can be shown coefficient-wise by using that (\ref{eq chinese remainder theorem}) is a ring isomorphism. Similarly, the other statements of the theorem can be proven coefficient-wise using the Chinese Remainder Theorem.
\end{proof}

The following codes will be useful thanks to Lemma \ref{lemma chinese remainder theorem}.

\begin{definition} \label{def skew-cyclic component code}
Let $ \pi_i : \mathcal{R}_1 \times \mathcal{R}_2 \times \cdots \times \mathcal{R}_t \longrightarrow \mathcal{R}_i $ denote the projection map onto $ \mathcal{R}_i $, for $ i = 1,2, \ldots, t $. Given a CSC code $ \mathcal{C} \subseteq \mathcal{R} $, we define its $ i $th \textit{skew-cyclic component} as 
$$ \mathcal{C}^{(i)} = \pi_i(\rho(\mathcal{C})) \subseteq \mathcal{R}_i = \frac{\mathcal{S}_i [z;\sigma]}{(z^N - 1)}, $$
which is a skew-cyclic code of length $ N $ over the field $ \mathcal{S}_i $ (recall that $ m_i(x) $ is irreducible) with field automorphism $ \sigma : \mathcal{S}_i \longrightarrow \mathcal{S}_i $, for $ i = 1,2, \ldots, t $. 
\end{definition}

We may now state and prove the main result of this subsection.

\begin{theorem} \label{th single generator}
If $ \mathcal{C} \subseteq \mathcal{R} $ is a CSC code, then it holds that
$$ \rho(\mathcal{C}) = \mathcal{C}^{(1)} \times \mathcal{C}^{(2)} \times \cdots \times \mathcal{C}^{(t)}. $$
In addition, there exist unique $ g(x,z), h(x,z) \in \mathcal{S} [z; \sigma] $ satisfying the following three properties: 
\begin{enumerate}
\item
Their projections onto $ \mathcal{S}_i [z; \sigma] $ are monic in $ z $, for $ i = 1,2, \ldots, t $.
\item
It holds that
$$ \mathcal{C} = (g(x,z)) $$
in the ring $ \mathcal{R} = \mathcal{S}[z; \sigma] / (z^N - 1) $.
\item
It holds that
$$ g(x,z) h(x,z) = h(x,z) g(x,z) = z^N - 1 $$
in the ring $ \mathcal{S}[z; \sigma] $.
\end{enumerate}
In particular, $ \mathcal{C} $ is a principal left ideal and $ \mathcal{R} $ and $ \mathcal{R}^\prime $ are PIRs. 

Furthermore, the images of $ g(x,z) $ and $ h(x,z) $ in $ \mathcal{S}_i [z;\sigma] $, denoted by $ g_i(x,z), h_i(x,z) \in \mathcal{S}_i [z;\sigma] $, are the minimal generator and check skew polynomials, respectively, of the $ i $th skew-cyclic component $ \mathcal{C}^{(i)} \subseteq \mathcal{R}_i $, for $ i = 1,2, \ldots, t $. In particular, it holds that
\begin{equation} \label{eq generator check pols from components}
\begin{split}
g(x,z) & = \sum_{i = 1}^t e_i(x) \widetilde{g}_i(x,z) = \sum_{i = 1}^t \widetilde{g}_i(x,z) e_i(x), \textrm{ and} \\
h(x,z) & = \sum_{i = 1}^t e_i(x) \widetilde{h}_i(x,z) = \sum_{i = 1}^t \widetilde{h}_i(x,z) e_i(x),
\end{split}
\end{equation}
in the ring $ \mathcal{S} [z; \sigma] $, for any $ \widetilde{g}_i(x,z), \widetilde{h}_i(x,z) \in \mathcal{S} [z;\sigma] $ such that their projections onto $ \mathcal{S}_i [z;\sigma] $ are $ g_i(x,z) $ and $ h_i(x,z) $, respectively. 
\end{theorem}
\begin{proof}
It follows directly from Lemma \ref{lemma chinese remainder theorem} (mainly (\ref{eq primitive idempotents are idempotents})) and the corresponding results for skew-cyclic codes (for instance, \cite[Lemma 1]{skewcyclic1} or \cite[Th. 1, 2]{ontheroots}).
\end{proof}

Theorem \ref{th single generator} motivates the following definition.

\begin{definition} \label{def minimal gen and check skew pol}
Let $ \mathcal{C} \subseteq \mathcal{R} $ be a CSC code, and let $ g(x,z), h(x,z) \in \mathcal{S} [z; \sigma] $ be as in Theorem \ref{th single generator}. We say that $ g(x,z) $ is the \textit{minimal generator skew polynomial} of $ \mathcal{C} $, and $ h(x,z) $ is the \textit{minimal check skew polynomial} of $ \mathcal{C} $.
\end{definition}

As expected, Definition \ref{def minimal gen and check skew pol} above recovers the classical definition of minimal generator and check skew polynomials when $ \ell = 1 $ \cite[Th. 1]{ontheroots}. However, in the classical cyclic case $ m = N = 1 $, it holds that
\begin{equation}
\begin{split}
g(x,z) & = \sum_{i \in I} e_i(x) + \sum_{j \in [t] \setminus I} e_j(x) ( z - 1 ), \textrm{ and} \\
h(x,z) & = \sum_{j \in [t] \setminus I} e_j(x) + \sum_{i \in  I} e_i(x) ( z - 1 ),
\end{split}
\label{eq min gen and check skew pol for classical case idemp}
\end{equation}
for a set $ I \subseteq [t] $. Hence the image of $ g(x,z) $ in $ \mathcal{R} $ coincides with the \textit{unique idempotent generator} of $ \mathcal{C} $ \cite[Th. 4.3.2]{pless}, which is $ \sum_{i \in I} e_i(x) \in \mathbb{F}[x] / (x^\ell - 1) $ \cite[Th. 4.3.8 (vi)]{pless}.

\subsection{Finding a Generator Matrix and the Dimension} \label{subsec consequences gen pol}

We next obtain a basis of a CSC code as a vector space from its minimal generator skew polynomial. The assumptions and notation will be as in Subsection \ref{subsec finding single gen pol}.

\begin{theorem} \label{th generator matrix}
Let $ \mathcal{C} \subseteq \mathcal{R} $ be a CSC code, and let $ g(x,z) \in \mathcal{S} [z; \sigma] $ be its minimal generator skew polynomial. Given $ \mathbf{c} \in \mathbb{F}^n $ and $ c(x,z) = \nu (\mathbf{c}) \in \mathcal{R} $, it holds that $ c(x,z) \in \mathcal{C} $ if, and only if, there exist coefficients
\begin{equation} \label{eq coeffs for gen matrix}
\lambda^{(i)}_{u,v} \in \mathbb{F},
\end{equation}
for $ u = 0,1, \ldots , d_i - 1 $, for $ v = 0,1, \ldots, k_i - 1 $, for $ i = 1,2, \ldots, t $, such that
\begin{equation} \label{eq linear comb for gen matrix}
c(x,z) = \sum_{i=1}^t \left( \sum_{v = 0}^{k_i - 1} \sum_{u = 0}^{d_i - 1} \lambda^{(i)}_{u,v} x^u z^v \right) e_i(x) g(x,z)
\end{equation}
inside the ring $ \mathcal{R} $, where 
$$ k_i = N - \deg_z(g_i(x,z)) \quad \textrm{and} \quad d_i = \deg_x(m_i(x)), $$
for $ i = 1,2, \ldots, t $, denoting by $ \deg_x(\cdot) $ the degree in $ x $, and similarly for other variables. Furthermore, the coefficients in (\ref{eq coeffs for gen matrix}) are unique among coefficients in $ \mathbb{F} $ satisfying (\ref{eq linear comb for gen matrix}). In particular, a basis of $ \mathcal{C} $ over $ \mathbb{F} $ is formed by the skew polynomials
$$ x^u z^v e_i(x) g(x,z) \in \mathcal{R}, $$
for $ u = 0,1, \ldots , d_i - 1 $, for $ v = 0,1, \ldots, k_i - 1 $, for $ i = 1,2, \ldots, t $.
\end{theorem}
\begin{proof}
By Theorem \ref{th single generator} and the corresponding result for skew-cyclic codes (see, e.g., \cite[Th. 1]{ontheroots}), it holds that 
$$ g_i(x,z), zg_i(x,z), \ldots, z^{k_i - 1}g_i(x,z) \in \mathcal{R}_i $$
form a basis of the skew-cyclic code $ \mathcal{C}^{(i)} \subseteq \mathcal{R}_i = \mathcal{S}_i [z; \sigma] /(z^N - 1) $ over the field $ \mathcal{S}_i $, for $ i = 1,2, \ldots, t $. Furthermore, we obtain a basis of $ \mathcal{S}_i $ over $ \mathbb{F} $ formed by
$$ 1, x, x^2, \ldots, x^{d_i - 1} \in \mathcal{S}_i, $$
for $ i = 1,2, \ldots, t $. Thus we deduce the result by using Theorem \ref{th single generator} and 
$$ \rho(e_i(x) g(x,z)) = \mathbf{e}_i g_i(x,z), $$
for $ i = 1,2, \ldots, t $.
\end{proof}

\begin{remark}
It may happen that $ k_i = 0 $ for some $ i = 1,2, \ldots, t $. This means that $ g_i(x,z) = z^N - 1 $ and $ \mathcal{C}^{(i)} = \{ \mathbf{0} \} $. This also means that the $ i $th term in the sum (\ref{eq linear comb for gen matrix}) does not exist.
\end{remark}

The first important consequence of the previous theorem is obtaining a generator matrix of a CSC code over its defining field, $ \mathbb{F} $.

\begin{corollary} \label{cor basis from gen skew pol}
Let $ \mathcal{C} \subseteq \mathcal{R} $ be a CSC code, and let $ g(x,z) \in \mathcal{S} [z; \sigma] $ be its minimal generator skew polynomial. Taking images in $ \mathcal{R} $, let
$$ e_i(x) g(x,z) = \left( \sum_{j = 0}^{N-1} \left( \sum_{h=0}^{\ell - 1} g^{(h)}_{i,j} x^h \right) z^j \right) \in \mathcal{R}, $$
and define
\begin{equation} \label{eq coordinates of basis from gen pol}
\begin{array}{rll}
\mathbf{g}_i & = (\mathbf{g}_i^{(0)}, \mathbf{g}_i^{(1)}, \ldots, \mathbf{g}_i^{(\ell - 1)}) & \in \mathbb{F}^n, \textrm{ where} \\
\mathbf{g}_i^{(h)} & = (g^{(h)}_{i,0}, g^{(h)}_{i,1}, \ldots, g^{(h)}_{i, N - 1}) & \in \mathbb{F}^N,
\end{array}
\end{equation}
for $  h = 0,1, \ldots, \ell - 1 $ and $ i = 1,2, \ldots, t $. Then the vectors
\begin{equation} \label{eq rows of generator matrix}
\varphi^u \left( \phi^v \left( \mathbf{g}_i \right) \right) \in \mathbb{F}^n,
\end{equation}
for $ u = 0,1, \ldots , d_i - 1 $, for $ v = 0,1, \ldots, k_i - 1 $, for $ i = 1,2, \ldots, t $, form a basis of $ \mathcal{C} \subseteq \mathbb{F}^n $ over $ \mathbb{F} $, and thus they form the rows of a generator matrix of $ \mathcal{C} $ over $ \mathbb{F} $.
\end{corollary}

The second important consequence is finding the dimension of a CSC code from its minimal generator skew polynomial.

\begin{corollary} \label{cor dimension from generator}
Let $ \mathcal{C} \subseteq \mathcal{R} $ be a CSC code. Let $ g(x,z) \in \mathcal{S} [z; \sigma] $ be its minimal generator skew polynomial, and let $ g_i(x,z) \in \mathcal{S}_i [z;\sigma] $ be its projection onto $ \mathcal{S}_i [z;\sigma] $, for $ i = 1,2, \ldots, t $. It holds that 
\begin{equation*}
\begin{split}
\dim_\mathbb{F}(\mathcal{C}) & = \sum_{i=1}^t \dim_\mathbb{F} \left( \mathcal{C}^{(i)} \right) \\
 & = \sum_{i = 1}^t \deg_x(m_i(x)) \left( N - \deg_z(g_i(x,z)) \right) \\
 & = n - \sum_{i=1}^t \deg_x(m_i(x)) \deg_z(g_i(x,z)) .
\end{split}
\end{equation*}
(Recall that $ 0 \leq \deg_z(g_i(x,z)) \leq N $ and $ \deg_z(g_i(x,z)) = N $ if, and only if, $ g_i(x,z) = z^N - 1 $, for $ i = 1,2, \ldots, t $.) 
\end{corollary}

\begin{remark}
Note that $ \deg_z(g_i(x,z)) $ is needed in Corollaries \ref{cor basis from gen skew pol} and \ref{cor dimension from generator}. It is left to the reader to verify that
$$ \deg_z(e_i(x) g(x,z)) = \deg_z(g_i(x,z)), $$
considering $ e_i(x) g(x,z) \in \mathcal{S} [z;\sigma] $ and $ g_i(x,z) \in \mathcal{S}_i [z;\sigma] $, for $ i = 1,2, \ldots, t $.
\end{remark}

We conclude by showing briefly which generator matrix one obtains from (\ref{eq rows of generator matrix}) in the cases $ m = N = 1 $ (classical cyclic codes) and $ \ell = 1 $ (skew-cyclic codes).

First, if $ \ell = 1 $, then $ t = 1 $, $ d_1 = 1 $, $ e_1(x) = 1 $, $ \mathcal{S} \cong \mathbb{F} $ naturally, and we may consider $ g(x,z) = g(z) \in \mathbb{F}[z; \sigma] $. In this case, the basis rows in (\ref{eq rows of generator matrix}) are
$$ \mathbf{g}, \phi(\mathbf{g}), \phi^2(\mathbf{g}), \ldots, \phi^{k-1}(\mathbf{g}) \in \mathbb{F}^N, $$
where $ \mathbf{g} \in \mathbb{F}^N $ is formed by the first $ N $ coefficients of $ g(z) $ in $ z $, and $ k = N - \deg_z(g(z)) $. Hence we recover the generator matrix of a skew-cyclic code as well as the well-known formula for its dimension (see, e.g., items 3 and 4 in \cite[Th. 1]{ontheroots}).

Finally, if $ m = N = 1 $, then as shown in (\ref{eq min gen and check skew pol for classical case idemp}), we have that
$$ g(x,z) = \sum_{i \in I} e_i(x) + \sum_{j \in [t] \setminus I} e_j(x) ( z - 1 ), $$
for a set $ I \subseteq [t] $. Then the rows of the generator matrix of $ \mathcal{C} $, given in (\ref{eq rows of generator matrix}), are
$$ \mathbf{e}_i, \varphi(\mathbf{e}_i), \varphi^2(\mathbf{e}_i), \ldots, \varphi^{d_i-1}(\mathbf{e}_i) \in \mathbb{F}^\ell, \textrm{ for } i \in I, $$
where $ \mathbf{e}_i \in \mathbb{F}^\ell $ is formed by the coefficients of $ e_i(x) \in \mathbb{F}[x] / (x^\ell - 1) $, for $ i = 1,2, \ldots, t $. This generator matrix corresponds to obtaining a classical cyclic code as a direct sum of the minimal cyclic codes generated by the primitive idempotents $ e_i(x) $, for $ i \in I $, and then appending the generator matrices of these minimal cyclic codes obtained from their corresponding idempotent as in \cite[Th. 4.3.6]{pless}.

\section{Sets of Roots of Cyclic-Skew-Cyclic Codes} \label{sec roots}

In this section, we provide a basic theory of defining sets of CSC codes that we will need later. Throughout this section, we will assume that $ N = m $, since this is the case of interest for defining sum-Rank BCH codes (Section \ref{sec SR BCH codes}).

\subsection{Defining Appropriate Evaluation Maps} \label{subsec evaluation maps}

In this subsection, we define the evaluation maps that we will consider in order to define roots of skew polynomials in $ \mathcal{R} \cong \mathcal{R}^\prime $. They will enable us to consider defining sets.

We start by revisiting the main definition of arithmetic evaluation of skew polynomials from \cite{lam, lam-leroy}, which reads as follows.

\begin{definition} [\textbf{\cite{lam, lam-leroy}}] \label{def skew polynomial evaluation}
Given a skew polynomial $ f(z) \in \mathbb{F}_{q^m}[z; \sigma] $ and a field element $ \alpha \in \mathbb{F}_{q^m} $, we define the evaluation of $ f(z) $ at $ \alpha $ as the unique element $ f(\alpha) \in \mathbb{F}_{q^m} $ such that there exists a skew polynomial $ g(z) \in \mathbb{F}_{q^m}[z; \sigma] $ satisfying that
$$ f(z) = g(z) (z - \alpha) + f(\alpha). $$
\end{definition}

Definition \ref{def skew polynomial evaluation} is consistent by the right Euclidean division property of the skew polynomial ring $ \mathbb{F}_{q^m}[z;\sigma] $ \cite{ore}. 

We now turn to the important class of linearized polynomials \cite[Sec. 3.4]{lidl}.

\begin{definition} \label{def linearized pols}
Let $ K $ be a finite field of characteristic $ p $. We say that a conventional polynomial $ G(y) \in K[y] $ is a linearized polynomial with coefficients in $ K $ and automorphism $ \sigma : K \longrightarrow K $, given by $ \sigma(a) = a^q $, for $ a \in K $, if it has the form
$$ G(y) = G_0 y + G_1 y^q + G_2 y^{q^2} + \cdots + G_d y^{q^d}, $$
where $ G_0, G_1, \ldots, G_d \in K $. We denote the set of such linearized polynomials by $ \mathcal{L}_q K [y] $. Finally, given a skew polynomial of the form
$$ f(z) = f_0 + f_1 z + f_2 z^2 + \cdots + f_d z^d \in K[z; \sigma], $$
where $ f_0, f_1, \ldots, f_d \in K $, we define its associated linearized polynomial as
$$ f^\sigma(y) = f_0 y + f_1 y^q + f_2 y^{q^2} + \cdots + f_d y^{q^d} \in \mathcal{L}_q K [y]. $$
\end{definition}

The set $ \mathcal{L}_q K [y] $ forms a ring with usual addition and with composition of maps as multiplication. As the reader may check, this ring is isomorphic as a ring to the skew polynomial ring $ K[z; \sigma] $ by mapping a skew polynomial to its associated linearized polynomial.

The following result connects the arithmetic evaluation of skew polynomials from Definition \ref{def skew polynomial evaluation} with the conventional evaluation of linearized polynomials. This result can be easily deduced from \cite[Lemma 2.4]{lam-leroy} and is also a particular case of \cite[Lemma 24]{linearizedRS}.

\begin{lemma}  \label{lemma connection arith eval lin eval}
Given a skew polynomial $ f(z) \in \mathbb{F}_{q^m}[z; \sigma] $, it holds that
$$ f( 1^\beta ) = f^\sigma(\beta) \beta^{-1}, $$
for all $ \beta \in \mathbb{F}_{q^m} \setminus \{ 0 \} $, where $ f^\sigma(\beta) $ is the conventional evaluation of the associated linearized polynomial $ f^\sigma(y) \in \mathcal{L}_q \mathbb{F}_{q^m}[y] $ in $ \beta $, and where we denote by
$$ 1^\beta = \sigma(\beta) \beta^{-1} $$
the conjugate of $ 1 $ with respect to $ \beta $ (see Subsection \ref{subsec revisiting lin RS codes}). 
\end{lemma} 

We may now define the evaluation maps on the roots of unity that we will use to provide defining sets of CSC codes.

\begin{definition} \label{def eval maps}
Given $ a \in \mathbb{F}_q $ and $ \beta \in \mathbb{F}_{q^m} \setminus \{ 0 \} $ such that $ a^\ell = 1 $, we define the \textit{evaluation maps}
\begin{equation*}
\begin{split}
{\rm Ev}_a : \frac{\mathbb{F}_{q^m}[x]}{(x^\ell - 1)} & \longrightarrow \frac{\mathbb{F}_{q^m}[x]}{(x - a)} \cong \mathbb{F}_{q^m} \textrm{ and} \\
{\rm Ev}^\sigma_\beta : \frac{\mathbb{F}_{q^m}[z; \sigma]}{(z^m - 1)} & \longrightarrow \frac{\mathbb{F}_{q^m}[z; \sigma]}{(z - \sigma(\beta)\beta^{-1})} \cong \mathbb{F}_{q^m}
\end{split}
\end{equation*}
as the natural projection maps onto the corresponding quotient left modules.
\end{definition}

Observe that the evaluation maps in Definition \ref{def eval maps} are well-defined since
$$ (x^\ell - 1) \subseteq (x-a) \quad \textrm{and} \quad (z^m - 1) \subseteq (z - \sigma(\beta)\beta^{-1}), $$
where the first inclusion follows from $ a^\ell = 1 $, and the second inclusion follows from $ \sigma^m(\beta) = \beta^{q^m} = \beta $ and Lemma \ref{lemma connection arith eval lin eval}. As expected, it holds that
$$ {\rm Ev}_a(f(x)) = f(a) \quad \textrm{and} \quad {\rm Ev}^\sigma_\beta (g(z)) = g(1^\beta), $$
for $ f(x) \in \mathbb{F}_{q^m}[x] / (x^\ell - 1) $ and $ g(z) \in \mathbb{F}_{q^m}[z; \sigma] / (z^m - 1) $. 

We will also find it useful to consider the following partial evaluation map and evaluation skew-cyclic codes.

\begin{definition} \label{def partial eval map}
Given $ a \in \mathbb{F}_q $ such that $ a^\ell = 1 $, we define the \textit{partial evaluation map}
$$ {\rm Ev}_{a,z} : \frac{\left( \frac{\mathbb{F}_{q^m}[x]}{(x^\ell - 1)} \right) [z;\sigma]}{(z^m - 1)} \longrightarrow \frac{\left( \frac{\mathbb{F}_{q^m}[x]}{(x - a)} \right) [z;\sigma]}{(z^m - 1)} \cong \frac{\mathbb{F}_{q^m}[z;\sigma]}{(z^m - 1)} $$
as the ring morphism given by
$$ {\rm Ev}_{a,z} \left( f_0(x) + f_1(x) z + \cdots + f_{m-1}(x) z^{m-1} \right) = $$
$$ f_0(a) + f_1(a) z + \cdots + f_{m-1}(a) z^{m-1}, $$
where $ f_0(x), f_1(x), \ldots, f_{m-1}(x) \in \mathbb{F}_{q^m}[x] / (x^\ell - 1) $. Finally, we will denote 
$$ f(a,z) = {\rm Ev}_{a,z} (f(x,z)) \in \frac{\mathbb{F}_{q^m}[z;\sigma]}{(z^m - 1)}, $$
\begin{flalign*}
\textrm{for } f(x,z) \in \frac{\left( \frac{\mathbb{F}_{q^m}[x]}{(x^\ell - 1)} \right) [z;\sigma]}{(z^m - 1)}. &&&
\end{flalign*}
We will use the same definition and notation for
$$ {\rm Ev}_{a,z} : \left( \frac{\mathbb{F}_{q^m}[x]}{(x^\ell - 1)} \right) [z;\sigma] \longrightarrow \left( \frac{\mathbb{F}_{q^m}[x]}{(x - a)} \right) [z;\sigma] \cong\mathbb{F}_{q^m}[z;\sigma] . $$
\end{definition}

It is left to the reader to prove that $ {\rm Ev}_{a,z} $ is a ring morphism. To this end, the reader might find it useful to use that $ \sigma(a) = a^q = a $ ($ a \in \mathbb{F}_q $ by assumption) to show that
$$ {\rm Ev}_{a,z} \left( f(x,z) g(x,z) \right) = f(a,z) g(a,z), $$
\begin{flalign*}
\textrm{for } f(x,z), g(x,z) \in \frac{\left( \frac{\mathbb{F}_{q^m}[x]}{(x^\ell - 1)} \right) [z;\sigma]}{(z^m - 1)}. &&&
\end{flalign*}

Assume that $ p $ does not divide $ \ell $ (i.e., the roots of $ x^\ell - 1 \in \mathbb{F}_p[x] $ are simple). Observe that, if $ a \in \mathbb{F}_q $ is such that $ a^\ell = 1 $, then 
$$ f(a,z) \in \frac{\mathbb{F}_{q_0^{m d_i}}[z;\sigma]}{(z^m - 1)}, \quad \textrm{for } f(x,z) \in \frac{\left( \frac{\mathbb{F}[x]}{(x^\ell - 1)} \right) [z;\sigma]}{(z^m - 1)}. $$
Recall that we are denoting $ d_i = \deg_x(m_i(x)) $ (Theorem \ref{th generator matrix}) and $ \mathbb{F} = \mathbb{F}_{q_0^m} $, where $ m_i(x) \in \mathbb{F}[x] $ is the irreducible component of $ x^\ell - 1 $ in $ \mathbb{F}[x] $ such that $ m_i(a) = 0 $ (see Subsection \ref{subsec finding single gen pol}). We may now give the following useful definition.

\begin{definition} \label{def skew-cyclic partial eval code}
Assume that $ p $ does not divide $ \ell $. Given a CSC code $ \mathcal{C} \subseteq \mathcal{R} $ and an element $ a \in \mathbb{F}_q $ such that $ a^\ell = 1 $, we define its \textit{evaluation skew-cyclic code} on $ a $ as 
$$ \mathcal{C}(a) = \{ c(a,z) \mid c(x,z) \in \mathcal{C} \} \subseteq \frac{\mathbb{F}_{q_0^{m d_i}}[z;\sigma]}{(z^m - 1)}, $$
which is a skew-cyclic code of length $ m $ over the field $ \mathbb{F}_{q_0^{m d_i}} $ with field automorphism $ \sigma : \mathbb{F}_{q_0^{m d_i}} \longrightarrow \mathbb{F}_{q_0^{m d_i}} $, where $ m_i(x) \in \mathbb{F}[x] $ is the irreducible component of $ x^\ell - 1 $ in $ \mathbb{F}[x] $ such that $ m_i(a) = 0 $, and where $ d_i = \deg_x (m_i(x)) $.
\end{definition} 

We have the following identification between $ \mathcal{C}(a) $ and $ \mathcal{C}^{(i)} $ (Definition \ref{def skew-cyclic component code}), whose connection is the fact that $ m_i(a) = 0 $. This also proves that $ \mathcal{C}(a) $ is a skew-cyclic code.

\begin{proposition} \label{prop partial eval gives isomorphism}
Assume that $ p $ does not divide $ \ell $. Let $ a \in \mathbb{F}_q $ be such that $ a^\ell = 1 $, and let $ m_i(x) \in \mathbb{F}[x] $ be the irreducible component of $ x^\ell - 1 $ in $ \mathbb{F}[x] $ such that $ m_i(a) = 0 $. Consider the partial evaluation maps defined with the following domains and codomains:
\begin{equation}
{\rm Ev}_{a,z} : \mathcal{S}_i [z;\sigma] = \left( \frac{\mathbb{F}[x]}{(m_i(x))} \right) [z;\sigma] \longrightarrow \mathbb{F}_{q_0^{m d_i}} [z;\sigma],
\label{eq partial eval map as ring isomorphism no quotient}
\end{equation}
\begin{equation}
{\rm Ev}_{a,z} : \mathcal{R}_i = \frac{\left( \frac{\mathbb{F}[x]}{(m_i(x))} \right) [z;\sigma]}{(z^m - 1)} \longrightarrow \frac{\mathbb{F}_{q_0^{m d_i}} [z;\sigma]}{(z^m - 1)},
\label{eq partial eval map as ring isomorphism}
\end{equation}
where $ d_i = \deg_x(m_i(x)) $. Then the maps (\ref{eq partial eval map as ring isomorphism no quotient}) and (\ref{eq partial eval map as ring isomorphism}) are ring isomorphisms, which moreover satisfy that
$$ {\rm Ev}_{a,z} \left( \mathcal{C}^{(i)} \right) = \mathcal{C}(a), $$
for any CSC code $ \mathcal{C} \subseteq \mathcal{R} $ (Definitions \ref{def skew-cyclic component code} and \ref{def skew-cyclic partial eval code}). In particular, it holds that
$$ \dim_\mathbb{F} \left( \mathcal{C}^{(i)} \right) = \dim_\mathbb{F} \left( \mathcal{C}(a) \right). $$
Furthermore, if $ g(x,z) \in \mathcal{S} [z; \sigma] $ is the minimal generator skew polynomial of $ \mathcal{C} $, and $ g_i(x,z) $ is its image in $ \mathcal{S}_i [z; \sigma] $ (the minimal generator skew polynomial of $ \mathcal{C}^{(i)} $ by Theorem \ref{th single generator}), then we have that
$$ g(a,z) = g_i(a,z) \in \mathbb{F}_{q_0^{md_i}} [z; \sigma], $$ 
which moreover is the minimal generator skew polynomial of $ \mathcal{C}(a) $.
\end{proposition}
\begin{proof}
The fact that the maps $ {\rm Ev}_{a,z} $ are well-defined and ring morphisms can be proven in the same way as for the map in Definition \ref{def partial eval map}. The fact that they are ring isomorphisms can be seen from the fact that the natural evaluation map
$$ {\rm Ev}_a : \frac{\mathbb{F}[x]}{(m_i(x))} \longrightarrow \mathbb{F}_{q_0^{m d_i}} $$
is a field isomorphism. The fact that $ g(a,z) = g_i(a,z) $ follows from (\ref{eq generator check pols from components}) and
$$ e_j(a) = \delta_{i,j}, $$
where $ \delta_{i,j} $ is the Kronecker delta, for $ j = 1,2, \ldots, t $, which in turn follows from (\ref{eq bezout identities}) and (\ref{eq primitive idempotent}). The rest of the statements follow directly from $ g(a,z) = g_i(a,z) $, the definitions and the fact that the maps (\ref{eq partial eval map as ring isomorphism no quotient}) and (\ref{eq partial eval map as ring isomorphism}) are ring isomorphisms. 
\end{proof}

What we will want from Proposition \ref{prop partial eval gives isomorphism} is the following important consequence. It follows directly from Lemma \ref{lemma chinese remainder theorem} and Proposition \ref{prop partial eval gives isomorphism}.

\begin{corollary} \label{cor total eval product map is isom}
Assume that $ p $ does not divide $ \ell $. Let $ x^\ell - 1 = m_1(x) m_2(x) \cdots m_t(x) $ be the irreducible decomposition of $ x^\ell - 1 $ in $ \mathbb{F}[x] $, and choose $ a_1, a_2, \ldots, a_t \in \mathbb{F}_q $ such that $ m_i(a_i) = 0 $, for $ i = 1,2, \ldots, t $. Denote $ \mathbf{a} = (a_1, a_2, \ldots, a_t) \in \mathbb{F}_q^t $. Then the map
$$ {\rm Ev}_{\mathbf{a},z} : \equalto{\frac{ \left( \frac{\mathbb{F}[x]}{(x^\ell - 1)} \right) [z; \sigma]}{(z^m - 1)}}{\mathcal{R}} \longrightarrow \frac{\mathbb{F}_{q_0^{m d_1}}[z; \sigma]}{(z^m - 1)} \times \cdots \times \frac{\mathbb{F}_{q_0^{m d_t}}[z; \sigma]}{(z^m - 1)}, $$
given by $ {\rm Ev}_{\mathbf{a},z}(f(x,z)) = (f(a_1,z), $ $ f(a_2,z), $ $ \ldots, $ $ f(a_t,z)) $, for $ f(x,z) \in \mathcal{R} $, is a ring isomorphism. In particular, for any CSC code $ \mathcal{C} \subseteq \mathcal{R} $, it holds that
\begin{equation*}
\begin{split}
{\rm Ev}_{\mathbf{a}, z} (\mathcal{C}) & = \mathcal{C}(a_1) \times \mathcal{C}(a_2) \times \cdots \times \mathcal{C}(a_t), \textrm{ and} \\
\dim_\mathbb{F} (\mathcal{C}) & = \sum_{i = 1}^t \dim_\mathbb{F} (\mathcal{C}(a_i)),
\end{split}
\end{equation*}
and, for a given $ c(x,z) \in \mathcal{R} $, the following are equivalent:
\begin{enumerate}
\item
$ c(x,z) \in \mathcal{C} $.
\item
$ c(a,z) \in \mathcal{C}(a) $, for all $ a \in \mathbb{F}_q $ such that $ a^\ell = 1 $.
\item
$ c(a_i,z) \in \mathcal{C}(a_i) $, for all $ i = 1,2, \ldots, t $.
\end{enumerate}
\end{corollary}

We conclude this subsection by defining total evaluation maps.

\begin{definition} \label{def total evaluation map}
Given $ a \in \mathbb{F}_q $ and $ \beta \in \mathbb{F}_{q^m} \setminus \{ 0 \} $ such that $ a^\ell = 1 $, we define the \textit{total evaluation map} 
$$ {\rm Ev}_{a,\beta} : \frac{\left( \frac{\mathbb{F}_{q^m}[x]}{(x^\ell - 1)} \right) [z;\sigma]}{(z^m - 1)} \longrightarrow \frac{\left( \frac{\mathbb{F}_{q^m}[x]}{(x - a)} \right) [z;\sigma]}{(z - \sigma(\beta)\beta^{-1})} \cong \mathbb{F}_{q^m} $$
as the composition map
$$ {\rm Ev}_{a,\beta} = {\rm Ev}^\sigma_\beta \circ {\rm Ev}_{a, z}. $$ 
\end{definition}

By Lemma \ref{lemma connection arith eval lin eval}, the total evaluation map is also given by
\begin{equation} \label{eq evaluation equation for diff roots}
\begin{split}
& {\rm Ev}_{a,\beta} \left( f_0(x) + f_1(x) z + \cdots + f_{m-1}(x) z^{m-1} \right) \\
& = f_0(a) + f_1(a) \sigma(\beta)\beta^{-1} + \cdots + f_{m-1}(a) \sigma^{m-1}(\beta)\beta^{-1} \\
 & = \left( f_0(a) \beta + f_1(a) \sigma(\beta) + \cdots + f_{m-1}(a) \sigma^{m-1}(\beta) \right) \beta^{-1} ,
\end{split}
\end{equation}
where $ f_0(x), f_1(x), \ldots, f_{m-1}(x) \in \mathbb{F}_{q^m}[x] / (x^\ell - 1) $. We will sometimes use the notation 
$$ f(a, 1^\beta) = {\rm Ev}_{a,\beta}(f(x,z)) , \quad \textrm{for } f(x,z) \in \frac{\left( \frac{\mathbb{F}_{q^m}[x]}{(x^\ell - 1)} \right) [z;\sigma]}{(z^m - 1)} . $$

\subsection{The Defining Set} \label{subsec roots and generators}

In this subsection, we show that the zeros of the minimal skew polynomial generator define a CSC code, as in the particular cases of classical cyclic codes \cite[Th. 4.4.2]{pless} and skew-cyclic codes \cite[Th. 2]{ontheroots}. We will use such defining sets to obtain the sum-rank BCH bound (Theorem \ref{th SR-BCH bound}) and a lower bound on the dimensions of sum-rank BCH codes (Theorem \ref{th lower bound}).

We start by the following definition.

\begin{definition} [\textbf{Defining Set}] \label{def defining set}
Given a CSC code $ \mathcal{C} \subseteq \mathcal{R} $ with minimal generator skew polynomial $ g(x,z) \in \mathcal{S} [z; \sigma] $, we define the \textit{defining set} of $ \mathcal{C} $ as
\begin{equation*}
\begin{split}
T_\mathcal{C} = \{ & (a, \beta) \in \mathbb{F}_q \times (\mathbb{F}_{q^m} \setminus \{0 \}) \mid \\
& a^\ell = 1, {\rm Ev}_{a,\beta}(g(x,z)) = 0 \} .
\end{split}
\end{equation*}
\end{definition}

As the name suggests, the defining set of a CSC code actually defines the CSC code. This is gathered in Theorem \ref{th defining set} below. We will need the following lemma, which follows from the \textit{product rule} \cite[Th. 2.7]{lam-leroy}, but can easily be proven from scratch.

\begin{lemma} \label{lemma product rule}
Given $ a \in \mathbb{F}_q $ and $ \beta \in \mathbb{F}_{q^m} \setminus \{ 0 \} $ such that $ a^\ell = 1 $, it holds that
$$ {\rm Ev}_{a,\beta} (f(x,z) g(x,z)) = 0 \quad \textrm{if } {\rm Ev}_{a,\beta} (g(x,z)) = 0, $$
\begin{flalign*}
\textrm{for } f(x,z), g(x,z) \in \frac{\left( \frac{\mathbb{F}_{q^m}[x]}{(x^\ell - 1)} \right) [z;\sigma]}{(z^m - 1)}. &&&
\end{flalign*}
\end{lemma}

\begin{theorem} \label{th defining set}
Assume that $ p $ does not divide $ \ell $. Given a CSC code $ \mathcal{C} \subseteq \mathcal{R} $, the following hold:
\begin{enumerate}
\item
Given $ c(x,z) \in \mathcal{R} $, it holds that $ c(x,z) \in \mathcal{C} $ if, and only if, $ c(a, 1^\beta) = 0 $, for all $ (a,\beta) \in T_\mathcal{C} $.
\item
Given another CSC code $ \widetilde{\mathcal{C}} \subseteq \mathcal{R} $, it holds that $ \mathcal{C} = \widetilde{\mathcal{C}} $ if, and only if, $ T_\mathcal{C} = T_{\widetilde{\mathcal{C}}} $.
\end{enumerate}
\end{theorem}
\begin{proof}
Item 2 follows immediately from item 1, hence we only prove item 1. 

First, if $ c(x,z) \in \mathcal{C} $, then $ c(a, 1^\beta) = 0 $ by Lemma \ref{lemma product rule}, since $ g(x,z) $ divides $ c(x,z) $ on the right and $ g(a, 1^\beta) = 0 $, for all $ (a,\beta) \in T_\mathcal{C} $.

Conversely, assume that $ c(a, 1^\beta) = 0 $, for all $ (a,\beta) \in T_\mathcal{C} $. Fix an element $ a \in \mathbb{F}_q $ such that $ a^\ell = 1 $. By the assumption on $ c(x,z) $ and the definition of $ T_\mathcal{C} $, it holds that $ c(a, 1^\beta) = 0 $, for all $ \beta \in \mathbb{F}_{q^m} \setminus \{ 0 \} $ such that $ g(a,1^\beta) = {\rm Ev}^\sigma_\beta (g(a,z)) = 0 $. By the corresponding result for skew-cyclic codes \cite[Th. 2]{ontheroots}, we have that $ c(a,z) \in \mathcal{C}(a) $, since $ g(a,z) $ is the minimal generator skew polynomial of $ \mathcal{C}(a) $ (Proposition \ref{prop partial eval gives isomorphism}). Thus, we have shown that $ c(a,z) \in \mathcal{C}(a) $, for all $ a \in \mathbb{F}_q $ such that $ a^\ell = 1 $. By Corollary \ref{cor total eval product map is isom}, we conclude that $ c(x,z) \in \mathcal{C} $, and we are done.
\end{proof}

We conclude by computing the dimension of a CSC code from its defining set. 

\begin{theorem} \label{th dimensions from defining set}
Assume that $ p $ does not divide $ \ell $. Let $ \mathcal{C} \subseteq \mathcal{R} $ be a CSC code. For $ a \in \mathbb{F}_q $ such that $ a^\ell = 1 $, define
\begin{equation}
\begin{split}
T_\mathcal{C}(a) & = \{ \beta \in \mathbb{F}_{q^m} \setminus \{ 0 \} \mid (a, \beta) \in T_\mathcal{C} \} \\
& = \{ \beta \in \mathbb{F}_{q^m} \setminus \{ 0 \} \mid {\rm Ev}^\sigma_\beta(g(a,z)) = 0 \} ,
\end{split}
\label{eq def defining sets of partial ev skew-cyclic codes}
\end{equation}
which satisfies that $ T_\mathcal{C}(a) \cup \{ 0 \} \subseteq \mathbb{F}_{q^m} $ is a vector space over $ \mathbb{F}_q $. Let $ x^\ell - 1 = m_1(x) m_2(x) \cdots m_t(x) $ be the irreducible decomposition of $ x^\ell - 1 $ in $ \mathbb{F}[x] $, and choose $ a_1, a_2, \ldots, a_t \in \mathbb{F}_q $ such that $ m_i(a_i) = 0 $, for $ i = 1,2, \ldots, t $. It holds that
\begin{equation*}
\begin{split}
\dim_\mathbb{F} (\mathcal{C}) & = \sum_{i=1}^t \deg_x (m_i(x)) \left( m - \dim_{\mathbb{F}_q} (T_\mathcal{C}(a_i) \cup \{ 0 \}) \right) \\
& = n - \sum_{i=1}^t \deg_x (m_i(x)) \dim_{\mathbb{F}_q} (T_\mathcal{C}(a_i) \cup \{ 0 \}) .
\end{split}
\end{equation*}
\end{theorem}
\begin{proof}
The fact that $ T_\mathcal{C}(a) \cup \{ 0 \} $ is a vector space over $ \mathbb{F}_q $, for $ a \in \mathbb{F}_q $ such that $ a^\ell = 1 $, follows directly from Lemma \ref{lemma connection arith eval lin eval}. 

Fix $ a \in \mathbb{F}_q $ such that $ a^\ell = 1 $. Let $ g(x,z) \in \mathcal{S}[z; \sigma] $ be the minimal generator skew polynomial of $ \mathcal{C} $. By Proposition \ref{prop partial eval gives isomorphism}, $ g(a,z) \in \mathbb{F}_{q_0^{m d_i}} [z; \sigma] $ is the minimal generator skew polynomial of $ \mathcal{C}(a) $. 

Consider the extended code over $ \mathbb{F}_{q^m} $
\begin{equation*}
\begin{split}
\mathcal{C}(a) \otimes \mathbb{F}_{q^m} = \{ & \lambda c(a,z) \mid \\
& \lambda \in \mathbb{F}_{q^m}, c(a,z) \in \mathcal{C}(a) \} \subseteq \frac{\mathbb{F}_{q^m} [z; \sigma]}{(z^m - 1)},
\end{split}
\end{equation*}
which is also a skew-cyclic code. Using item 5 in \cite[Th. 2]{ontheroots}, we deduce that $ g(a,z) \in \mathbb{F}_{q^m}[z; \sigma] $ is also the minimal generator skew polynomial of $ \mathcal{C}(a) \otimes \mathbb{F}_{q^m} $, since it divides $ z^m-1 $ on the right. Furthermore, using Lemma \ref{lemma connection arith eval lin eval}, we can see that $ T_\mathcal{C}(a) \cup \{ 0 \} $ is the set of roots of the linearized polynomial associated to $ g(a,z) \in \mathbb{F}_{q^m} [z;\sigma] $ (Definition \ref{def linearized pols}). Thus, applying item 2 in \cite[Th. 2]{ontheroots} on the skew-cyclic code $ \mathcal{C}(a) \otimes \mathbb{F}_{q^m} $, we deduce that its dimension over $ \mathbb{F}_{q^m} $ is given by
\begin{equation} \label{eq dim-roots eq 1}
\dim_{\mathbb{F}_{q^m}} \left( \mathcal{C}(a) \otimes \mathbb{F}_{q^m} \right) = m - \dim_{\mathbb{F}_q} \left( T_\mathcal{C}(a) \cup \{ 0 \} \right).
\end{equation}
The reader can also verify the identities
\begin{equation} \label{eq dim-roots eq 2}
\begin{array}{lcr}
\dim_{\mathbb{F}_{q^m}} \left( \mathcal{C}(a) \otimes \mathbb{F}_{q^m} \right) & = & \dim_{\mathbb{F}_{q_0^{m d_i}}} \left( \mathcal{C}(a) \right), \\
\dim_\mathbb{F} \left( \mathcal{C}(a) \right) & = & \deg_x (m_i(x)) \dim_{\mathbb{F}_{q_0^{m d_i}}} \left( \mathcal{C}(a) \right).
\end{array}
\end{equation}
Combining Corollary \ref{cor total eval product map is isom} with (\ref{eq dim-roots eq 1}) and (\ref{eq dim-roots eq 2}) for $ a_1, a_2, \ldots, a_t $, we conclude that
\begin{equation*}
\begin{split}
\dim_\mathbb{F} \left( \mathcal{C} \right) & = \sum_{i = 1}^t \dim_\mathbb{F} (\mathcal{C}(a_i)) \\
& = \sum_{i=1}^t \deg_x (m_i(x)) \left( m - \dim_{\mathbb{F}_q} (T_\mathcal{C}(a_i) \cup \{ 0 \}) \right) ,
\end{split}
\end{equation*}
and we are done.
\end{proof}

To conclude, we discuss the cases $ \ell = 1 $ and $ m = 1 $. 

First, in the skew-cyclic case $ \ell = 1 $, Definition \ref{def defining set} coincides with \cite[Def. 2]{ontheroots} after removing the redundant unique root of unity $ a = 1 $. Note that \cite[Def. 2]{ontheroots} uses the associated linearized polynomial (Definition \ref{def linearized pols}) of the minimal generator skew polynomial. Finally, Theorems \ref{th defining set} and \ref{th dimensions from defining set} recover items 3 and 2 in \cite[Th. 2]{ontheroots}, respectively.

We now turn to the classical cyclic case $ m = 1 $. With notation as in (\ref{eq min gen and check skew pol for classical case idemp}), the reader may verify (using that $ e_j(a) = \delta_{i,j} $ if $ m_i(a) = 0 $) that 
\begin{equation*}
\begin{split}
T_\mathcal{C} = \{ & (a, \beta) \in \mathbb{F}_q \times (\mathbb{F}_q \setminus \{ 0 \} ) \mid \\
& m_i(a) = 0, \textrm{ for some } i \notin I \} .
\end{split}
\end{equation*}
In other words, after removing the second component $ \beta \in \mathbb{F}_q \setminus \{ 0 \} $, $ T_\mathcal{C} $ is exactly the set of roots of unity that are roots of the idempotent generator of $ \mathcal{C} $. Such set is precisely the set of roots of the minimal generator polynomial of $ \mathcal{C} $. Thus Definition \ref{def defining set} coincides with the standard definition of defining set of a classical cyclic code \cite[Sec. 7.5]{macwilliamsbook} \cite[Sec. 4.4]{pless}. Furthermore, Theorems \ref{th defining set} and \ref{th dimensions from defining set} recover items (iii) and (iv) in \cite[Th. 4.4.2]{pless}.

\section{Cyclic-Skew-Cyclic Linearized Reed-Solomon Codes} \label{sec CSC LRS codes}

In this subsection, we revisit linearized Reed-Solomon codes \cite{linearizedRS, caruso} and provide one of their subfamilies formed by CSC codes. They will be crucial for proving the sum-rank BCH bound (Theorem \ref{th SR-BCH bound}) and defining sum-rank BCH codes (Definition \ref{def SR-BCH codes}).

Recall that $ \mathbb{F} = \mathbb{F}_{q_0^m} $, by definition (\ref{eq main finite field}), and that $ q = q_0^s $ for some positive integer $ s $. In this section, we will only work with the finite-field extension $ \mathbb{F}_q \subseteq \mathbb{F}_{q^m} $.

\subsection{Revisiting Linearized RS Codes} \label{subsec revisiting lin RS codes}

As before, we also consider the length partition $ n = \ell N $ (Subsection \ref{subsec sum-rank metric}). However, we need to assume that $ 1 \leq \ell \leq q-1 $ and $ 1 \leq N \leq m $. Therefore $ n \leq (q-1) m $.

Recall that we consider $ \sigma : \mathbb{F}_{q^m} \longrightarrow \mathbb{F}_{q^m} $ given by $ \sigma(a) = a^q $, for all $ a \in \mathbb{F}_{q^m} $. We need to define linear operators as in \cite[Def.~20]{linearizedRS}. 

\begin{definition}[\textbf{Linear operators \cite{linearizedRS}}] \label{def linearized operators}
Fix $ a \in \mathbb{F}_{q^m} $, and define its $ i $th norm as $ N_i(a)
= \sigma^{i-1}(a) \cdots \sigma(a)a $, for $ i \in \mathbb{N} $. Define the $ \mathbb{F}_q
$-linear operator $ \mathcal{D}_a^i : \mathbb{F}_{q^m} \longrightarrow
\mathbb{F}_{q^m} $ by
$$ \mathcal{D}_a^i(b) = \sigma^i(b) N_i(a) , $$
for $ b \in \mathbb{F}_{q^m} $ and $ i \in \mathbb{N} $. 
\end{definition}

We say that $ a,b \in \mathbb{F}_{q^m} $ are \textit{conjugate} (with respect to $ \mathbb{F}_{q^m}[z;\sigma] $) if there exists $ c \in \mathbb{F}_{q^m} \setminus \{ 0 \} $ such that $ b = \sigma(c)c^{-1} a $ (see \cite{lam} \cite[Eq.~(2.5)]{lam-leroy}). Let $ \mathcal{A} = \{ a_0, a_1, \ldots, a_{\ell - 1} \} \subseteq \mathbb{F}_{q^m} \setminus \{ 0 \} $ be a set of $ \ell $ pair-wise non-conjugate elements. Let $ \mathcal{B}_i = \{ \beta_{i,0}, \beta_{i,1}, \ldots, \beta_{i, N-1} \} \subseteq \mathbb{F}_{q^m} $ be a set of $ N $ elements of $ \mathbb{F}_{q^m} $ that are linearly independent over $ \mathbb{F}_q $, for $ i = 0,1, \ldots, \ell -1 $. Denote $ \mathcal{B} = (\mathcal{B}_0, \mathcal{B}_1, \ldots, \mathcal{B}_{\ell - 1}) $ and define the matrix 
\begin{equation}
D(\mathcal{A}, \mathcal{B}) = (D_0 | D_1 | \ldots | D_{\ell - 1}) \in \mathbb{F}_{q^m}^{k \times n},
\label{eq gen matrix lin RS codes general}
\end{equation}
where
\begin{equation*}
D_i = \left( \begin{array}{cccc}
\beta_{i,0} & \beta_{i,1} & \ldots & \beta_{i,N-1} \\
\mathcal{D}_{a_i} \left( \beta_{i,0} \right) & \mathcal{D}_{a_i} \left( \beta_{i,1} \right) & \ldots & \mathcal{D}_{a_i} \left( \beta_{i,N-1} \right) \\
\mathcal{D}_{a_i}^2 \left( \beta_{i,0} \right) & \mathcal{D}_{a_i}^2 \left( \beta_{i,1} \right) & \ldots & \mathcal{D}_{a_i}^2 \left( \beta_{i,N-1} \right) \\
\vdots & \vdots & \ddots & \vdots \\
\mathcal{D}_{a_i}^{k-1} \left( \beta_{i,0} \right) & \mathcal{D}_{a_i}^{k-1} \left( \beta_{i,1} \right) & \ldots & \mathcal{D}_{a_i}^{k-1} \left( \beta_{i,N-1} \right) \\
\end{array} \right),
\end{equation*}
for $ i = 0,1, \ldots, \ell - 1 $, for $ k = 1,2, \ldots, n $. Then the following definition is a particular case of \cite[Def.~31]{linearizedRS}.

\begin{definition} [\textbf{Linearized Reed-Solomon codes \cite{linearizedRS}}] \label{def lin RS codes}
For $ k = 1,2, \ldots, n $, we define the linearized Reed-Solomon code of dimension $ k $, set of pair-wise non-conjugate elements $ \mathcal{A} $ and linearly independent sets $ \mathcal{B} $, as the linear code $
\mathcal{C}^\sigma_{k}(\mathcal{A}, \mathcal{B}) \subseteq \mathbb{F}_{q^m}^n $ with generator matrix $ D(\mathcal{A}, \mathcal{B}) \in  \mathbb{F}_{q^m}^{k \times n} $ as in (\ref{eq gen matrix lin RS codes general}).
\end{definition}

The following result is \cite[Th. 4]{linearizedRS} and states that linearized Reed-Solomon codes attain equality in (\ref{eq singleton bound}), expressing size in terms of dimension.

\begin{proposition}[\textbf{\cite{linearizedRS}}] \label{prop linRS are MSRD}
For $ k = 1,2, \ldots, n $, the linearized Reed-Solomon code $ \mathcal{C}^\sigma_{k}(\mathcal{A}, \mathcal{B}) \subseteq \mathbb{F}_{q^m}^n $ in Definition \ref{def lin RS codes} is a $ k $-dimensional $ \mathbb{F}_{q^m} $-linear MSRD code for the finite-field extension $ \mathbb{F}_q \subseteq \mathbb{F}_{q^m} $ and length partition $ n = \ell N $ as in (\ref{eq sum-rank length partition}). That is, it satisfies that
$$ {\rm d}_{SR}(\mathcal{C}^\sigma_{k}(\mathcal{A}, \mathcal{B})) = n - k + 1. $$
\end{proposition}

As observed in \cite[Sec. 3]{linearizedRS} and \cite[Subsec. IV-A]{secure-multishot}, linearized Reed-Solomon codes recover (generalized) Reed-Solomon codes \cite{reed-solomon} when $ m = N = 1 $, and they recover Gabidulin codes \cite{gabidulin, roth} when $ \ell=1 $. These are the cases when the sum-rank metric particularizes to the Hamming metric and the rank metric, respectively.

\subsection{Non-Conjugate Roots of Unity} \label{subsec lin RS codes conj repr lie in base field}

The next lemma is a general result on $ \ell $th roots of unity and conjugacy with respect to $ \mathbb{F}_{q^m}[z; \sigma] $, but we have not been able to find it in the literature. It will enable us to obtain linearized Reed-Solomon codes that are also CSC codes (Subsection \ref{subsec finding CSC lin RS codes}), which in turn will allow us to obtain Sum-Rank BCH codes (Subsection \ref{subsec SR BCH bound}). 

\begin{lemma} \label{lemma non-conj over F_q}
Recall that we are assuming throughout the manuscript that $ \mathbb{F}_q $ contains all $ \ell $th roots of unity. Assume also that there are exactly $ \ell $ distinct $ \ell $th roots of unity, that is, $ p $ does not divide $ \ell $. Then the following conditions are equivalent:
\begin{enumerate}
\item
The roots of $ x^\ell - 1 \in \mathbb{F}_p[x] $ are pair-wise non-conjugate with respect to $ \mathbb{F}_{q^m}[z; \sigma] $ (see Subsection \ref{subsec revisiting lin RS codes}).
\item
$ \ell $ and $ m $ are coprime.
\end{enumerate}
\end{lemma}
\begin{proof}
By assumption, we have $ \ell $ distinct roots of $ x^\ell - 1 $, all inside $ \mathbb{F}_q $. Let $ a \in \mathbb{F}_q $ be a primitive $ \ell $th root of unity, meaning that all the $ \ell $th roots of unity are given by $ a^0, a^1, \ldots, a^{\ell - 1} $. Such primitive roots always exist \cite[Page 105]{pless} \cite[Sec. 2.4]{lidl} (see also the next subsection). By Hilbert's Theorem 90 \cite[Ex. 2.33]{lidl}, there exist two distinct $ \ell $th roots of unity $ a^i $ and $ a^j $ that are conjugate if, and only if,
$$ a^{im} = (a^i)^{\frac{q^m-1}{q-1}} = (a^j)^{\frac{q^m-1}{q-1}} = a^{jm}, $$ 
where $ 0 \leq i < j \leq \ell - 1 $. Now, since $ a $ is a primitive $ \ell $th root of unity, such a condition may happen if, and only if, $ \ell $ and $ m $ share a common factor. 
\end{proof}

\begin{remark} \label{remark q-1 and (q^m-1)/(q-1)}
Note that, as a particular case of Lemma \ref{lemma non-conj over F_q}, the $ q-1 $ elements in $ \mathbb{F}_q \setminus \{ 0 \} $ are pair-wise non-conjugate with respect to $ \mathbb{F}_{q^m}[z;\sigma] $ if, and only if, $ q-1 $ and $ m $ are coprime (note that the assumptions in Lemma \ref{lemma non-conj over F_q} above are trivially satisfied for $ \ell = q-1 $). 
\end{remark}

\begin{remark}
Applying the definition of conjugacy directly (see Subsection \ref{subsec revisiting lin RS codes}), without making use of Hilbert's Theorem 90, the reader may arrive at the fact that item 1 in Lemma \ref{lemma non-conj over F_q} is also equivalent to $ \ell $ and $ \frac{q^m - 1}{q-1} $ being coprime. Indeed, this latter condition is equivalent to $ \ell $ and $ m $ being coprime, under the assumption that $ \ell $ divides $ q-1 $ (i.e., $ \mathbb{F}_q $ contains all $ \ell $th roots of unity), as in Lemma \ref{lemma non-conj over F_q} above. To see this, note that
\begin{equation*}
\begin{split}
\frac{q^m-1}{q-1} - m & = \sum_{i=1}^{m-1} (q^i-1) \\
& = \left( \sum_{i=1}^{m-1} \frac{q^i-1}{q-1} \right) (q-1)
\end{split}
\end{equation*}
is divisible by $ \ell $ if $ q-1 $ is divisible by $ \ell $. In such a case, a factor of $ \ell $ divides $ \frac{q^m-1}{q-1} $ if, and only if, it divides $ m $. 
\end{remark}

\subsection{Finding CSC Linearized RS Codes} \label{subsec finding CSC lin RS codes}

Combining Definition \ref{def lin RS codes} and Lemma \ref{lemma non-conj over F_q}, we will describe a subfamily of linearized Reed-Solomon codes formed by CSC codes, which in addition recovers classical cyclic (generalized) Reed-Solomon codes and skew-cyclic Gabidulin codes when setting $ m = N = 1 $ and $ \ell = 1 $, respectively.

In this subsection, we assume that: (1) $ N = m $; (2) $ \ell $ and $ m $ are coprime; (3) $ \ell $ and $ q $ are coprime; and (4) $ x^\ell - 1 $ has all of its $ \ell $ distinct roots in $ \mathbb{F}_q $ ($ q = q_0^s $), i.e., $ \ell $ divides $ q-1 $. In the classical cyclic case $ m = 1 $, condition 2 is trivially satisfied, whereas in the skew-cyclic case $ \ell = 1 $, conditions 2, 3 and 4 are all trivially satisfied. Note that condition 2 is satisfied whether $ \ell = 1 $ or $ m = 1 $.

Let $ a \in \mathbb{F}_q $ be a \textit{primitive $ \ell $th root of unity}, meaning that the set of roots of $ x^\ell - 1 $ is given by
$$ \mathcal{A} = \left\lbrace a^0, a^1, a^2, \ldots, a^{\ell - 1} \right\rbrace \subseteq \mathbb{F}_q \setminus \{ 0 \} . $$
Such a primitive root always exists \cite[Page 105]{pless} \cite[Sec. 2.4]{lidl}. Since we are assuming that $ \ell $ is coprime with both $ q $ and $ m $, the set $ \mathcal{A} $ is formed by $ \ell $ distinct and pair-wise non-conjugate elements by Lemma \ref{lemma non-conj over F_q}.

Next, let $ \beta \in \mathbb{F}_{q^m} $ be a \textit{normal element} of the extension $ \mathbb{F}_q \subseteq \mathbb{F}_{q^m} $. In other words,
$$ \left\lbrace \beta, \sigma(\beta), \sigma^2(\beta), \ldots, \sigma^{m-1}(\beta) \right\rbrace \subseteq \mathbb{F}_{q^m} $$
forms a basis of $ \mathbb{F}_{q^m} $ over $ \mathbb{F}_q $. Recall that normal elements exist for any extension of finite fields \cite[Th. 3.73]{lidl}. We fix an integer $ b \geq 0 $. For $ i = 0,1, \ldots, \ell - 1 $, note that the element $ \beta a^{bi} \in \mathbb{F}_{q^m} $ is also normal. Denote the corresponding basis by
\begin{equation}
\mathcal{B}_i = \left\lbrace \beta a^{bi}, \sigma(\beta) a^{bi}, \ldots, \sigma^{m-1}(\beta) a^{bi} \right\rbrace \subseteq \mathbb{F}_{q^m}
\label{eq normal basis for SR-BCH codes}
\end{equation}
(recall that $ a^{bi} \in \mathbb{F}_q $, thus $ \sigma(a^{bi}) = a^{bi} $). We define $ \mathcal{B} = (\mathcal{B}_0, \mathcal{B}_1, \ldots, \mathcal{B}_{\ell - 1}) $.

For $ k = 1,2, \ldots, n $, the corresponding linearized Reed-Solomon code $
\mathcal{C}^\sigma_{k}(\mathcal{A}, \mathcal{B}) \subseteq \mathbb{F}_{q^m}^n $ (Definition \ref{def lin RS codes}) is given by the generator matrix $ D(\mathcal{A}, \mathcal{B}) = (D_0 | D_1 | \ldots | D_{\ell - 1}) \in \mathbb{F}_{q^m}^{k \times n} $, where $ n = \ell m $ and $ D_i \in \mathbb{F}_{q^m}^{k \times m} $ is given by
\begin{equation*}
\scalebox{0.90}{$
\left( \begin{array}{cccc}
 \beta a^{bi} & \sigma(\beta) a^{bi} & \ldots & \sigma^{m-1}(\beta) a^{bi} \\
 \sigma(\beta) a^{(b+1)i} & \sigma^2(\beta) a^{(b+1)i} & \ldots & \beta a^{(b+1)i} \\ 
 \sigma^2(\beta) a^{(b+2)i} & \sigma^3(\beta) a^{(b+2)i} & \ldots & \sigma(\beta) a^{(b+2)i} \\
 \vdots & \vdots & \ddots & \vdots \\
 \sigma^{k-1}(\beta) a^{(b+k-1)i} & \sigma^k(\beta) a^{(b+k-1)i} & \ldots & \sigma^{k-2}(\beta) a^{(b+k-1)i} \\
\end{array} \right),
$}
\end{equation*}
for all $ i = 0, 1, \ldots, \ell - 1 $.

We conclude with the main result of this section, whose proof is left to the reader.

\begin{theorem} \label{th CSC linearized RS codes}
For $ k = 1,2, \ldots, n $, the linearized Reed-Solomon code $
\mathcal{C}^\sigma_{k}(\mathcal{A}, \mathcal{B}) \subseteq \mathbb{F}_{q^m}^n $ as above is a CSC code in $ \mathbb{F}_{q^m}^n $ with field automorphism $ \sigma : \mathbb{F}_{q^m} \longrightarrow \mathbb{F}_{q^m} $ (Definition \ref{def CSC codes}).
\end{theorem}

\section{Sum-Rank BCH Codes} \label{sec SR BCH codes}

Throughout this section, we will make the same assumptions as in the beginning of Subsection \ref{subsec finding CSC lin RS codes}, that is: (1) $ N = m $; (2) $ \ell $ and $ m $ are coprime; (3) $ \ell $ and $ q $ are coprime; and (4) $ \ell $ divides $ q-1 $, i.e., $ x^\ell - 1 $ has all of its $ \ell $ distinct roots in $ \mathbb{F}_q $ ($ q = q_0^s $).

\subsection{The Sum-Rank BCH Bound Leading to SR-BCH Codes} \label{subsec SR BCH bound}

In this subsection, we provide a lower bound on the minimum sum-rank distance of CSC codes based on their defining set (Definition \ref{def defining set}), which will allow us to define Sum-Rank BCH codes (Definition \ref{def SR-BCH codes}). 

In order to prove our bound, we need the following lemma, which is \cite[Th. 7]{universal-lrc}. Recall from Subsection \ref{subsec sum-rank metric} that we denote by $ {\rm d}_{SR} $ and $ {\rm d}_{SR}^0 $ the sum-rank metrics over the finite-field extensions $ \mathbb{F}_q \subseteq \mathbb{F}_{q^m} $ and $ \mathbb{F}_{q_0} \subseteq \mathbb{F}_{q_0^m} $, respectively, for the length partition $ n = \ell m $ as in (\ref{eq sum-rank length partition}) for $ N = m $. 

\begin{lemma} [\textbf{\cite{universal-lrc}}] \label{lemma min sum-rank distance subextension subcodes}
For a code $ \mathcal{C} \subseteq \mathbb{F}_{q^m}^n $ (we do not assume that it is linear), it holds that
\begin{equation}
{\rm d}_{SR}^0 \left( \mathcal{C} \cap \mathbb{F}^n \right) \geq {\rm d}_{SR}(\mathcal{C}).
\label{eq subext subcode inherits sum-rank distance}
\end{equation}
\end{lemma}

\begin{theorem} [\textbf{Sum-Rank BCH bound}] \label{th SR-BCH bound}
Let $ a \in \mathbb{F}_q $ be a primitive $ \ell $th root of unity and let $ \beta \in \mathbb{F}_{q^m} $ be a normal element of the extension $ \mathbb{F}_q \subseteq \mathbb{F}_{q^m} $ (see Subsection \ref{subsec finding CSC lin RS codes}). Let $ \mathcal{C} \subseteq \mathbb{F}^n $ be a CSC code. If the defining set $ T_\mathcal{C} $ of $ \mathcal{C} $ (Definition \ref{def defining set}) contains the consecutive pairs in $ \mathbb{F}_q \times ( \mathbb{F}_{q^m} \setminus \{ 0 \} ) $,
\begin{equation}
\left( a^b, \beta \right), \left( a^{b+1}, \sigma(\beta) \right), \ldots, \left( a^{b + \delta - 2}, \sigma^{\delta - 2}(\beta) \right),
\label{eq pairs for SR-BCH codes}
\end{equation}
for integers $ b \geq 0 $ and $ 2 \leq \delta \leq n $, then it holds that
$$ {\rm d}_{SR}^0 (\mathcal{C}) \geq \delta . $$
\end{theorem}
\begin{proof}
If $ \mathcal{C} $ satisfies the hypothesis, then by Theorem \ref{th defining set} and equation (\ref{eq evaluation equation for diff roots}), it holds that 
\begin{equation}
\mathcal{C} \subseteq \left( \mathcal{C}^\sigma_{\delta - 1}(\mathcal{A}, \mathcal{B})^\perp \right) \cap \mathbb{F}^n ,
\label{eq proof BCH bound 1}
\end{equation}
where $ \mathcal{C}^\sigma_{\delta - 1}(\mathcal{A}, \mathcal{B}) $ is the $ (\delta - 1) $-dimensional linearized Reed-Solomon code (Definition \ref{def lin RS codes}) with $ \mathcal{A} = \{  1, a, a^2, $ $ \ldots, $ $ a^{\ell - 1} \} $ and $ \mathcal{B} = (\mathcal{B}_0, \mathcal{B}_1, \ldots, \mathcal{B}_{\ell-1}) $ given by (\ref{eq normal basis for SR-BCH codes}). 

By \cite[Th. 5]{gsrws}, the dual $ \mathcal{C}^\sigma_{\delta - 1}(\mathcal{A}, \mathcal{B})^\perp \subseteq \mathbb{F}_{q^m}^n $ is also an MSRD code, hence
\begin{equation}
{\rm d}_{SR} \left( \mathcal{C}^\sigma_{\delta - 1}(\mathcal{A}, \mathcal{B})^\perp \right) = n - (n - \delta + 1) + 1 = \delta.
\label{eq proof BCH bound 2}
\end{equation}

Combining (\ref{eq proof BCH bound 1}), (\ref{eq proof BCH bound 2}) and Lemma \ref{lemma min sum-rank distance subextension subcodes}, we conclude that $ {\rm d}_{SR}^0 (\mathcal{C}) \geq \delta $.
\end{proof}

This bound recovers the well-known BCH bound \cite[Sec. 7.6, Th. 8]{macwilliamsbook} \cite[Th. 4.5.3]{pless} (originally \cite{bose, hocquenghem}) in the classical cyclic case $ m = 1 $ and its rank-metric version \cite[Prop. 1]{skewcyclic3} in the skew-cyclic case $ \ell = 1 $. Note that \cite[Prop. 1]{skewcyclic3} is given for lengths $ N \geq m $, whereas here we only consider $ N = m $.

We may now define Sum-Rank BCH codes with prescribed distance.

\begin{definition} [\textbf{Sum-Rank BCH codes}] \label{def SR-BCH codes}
Let $ a \in \mathbb{F}_q $ be a primitive $ \ell $th root of unity and let $ \beta \in \mathbb{F}_{q^m} $ be a normal element of the finite-field extension $ \mathbb{F}_q \subseteq \mathbb{F}_{q^m} $ (see Subsection \ref{subsec finding CSC lin RS codes}). Fix integers $ b \geq 0 $ and $ 2 \leq \delta \leq n $. We define the corresponding Sum-Rank BCH code (or SR-BCH code for short) with prescribed distance $ \delta $ as
$$ \mathcal{C}_\delta(a^b,\beta) = \left( \mathcal{C}^\sigma_{\delta - 1}(\mathcal{A}, \mathcal{B})^\perp \right) \cap \mathbb{F}^n , $$
where $ \mathcal{C}^\sigma_{\delta - 1}(\mathcal{A}, \mathcal{B}) \subseteq \mathbb{F}_{q^m}^n $ is as in Definition \ref{def lin RS codes}, for $ \mathcal{A} = \{  1, a, $ $ a^2, $ $ \ldots, $ $ a^{\ell - 1} \} $ and $ \mathcal{B} $ $ = $ $ (\mathcal{B}_0, \mathcal{B}_1, \ldots, \mathcal{B}_{\ell-1}) $ given by $ \mathcal{B}_i = \{ \beta a^{bi}, \sigma(\beta) a^{bi}, \ldots, \sigma^{m-1}(\beta) a^{bi} \} $, for $ i = 0,1, \ldots, \ell -1 $.
\end{definition}

SR-BCH codes recover classical BCH codes \cite[Sec. 7.6]{macwilliamsbook} \cite[Sec. 5.1]{pless} when $ m = 1 $ (in that case, $ \mathcal{C}_\delta(a^b,\beta) = \mathcal{C}_\delta(a^b,1) $ for any $ \beta \in \mathbb{F}_{q^m} \setminus \{ 0 \} $), and they recover rank-metric BCH codes \cite[Page 272]{skewcyclic3} \cite[Def. 7]{ontheroots} when $ \ell = 1 $ for the code length $ N = m $. The latter are full-length skew-cyclic Gabidulin codes \cite[Sec. 5.2]{ontheroots} if $ s = 1 $. Moreover, SR-BCH codes form a subfamily of sum-rank alternant codes \cite[Def. 12]{universal-lrc}, which in turn recover classical alternant codes \cite[Ch. 12]{macwilliamsbook} when $ m = 1 $. 

The motivation for defining SR-BCH codes as in Definition \ref{def SR-BCH codes} is to obtain the largest CSC code (hence hopefully having maximum possible code size) for a prescribed distance in view of Theorem \ref{th SR-BCH bound}. This is shown in the following result.

\begin{proposition} \label{prop SR-BCH are largest containing pairs}
With assumptions and notation as in Definition \ref{def SR-BCH codes}, the SR-BCH code $ \mathcal{C}_\delta(a^b,\beta) \subseteq \mathbb{F}^n $ is a CSC code. Moreover, it is the largest CSC code in $ \mathbb{F}^n $, with respect to set inclusion, whose defining set contains the pairs in (\ref{eq pairs for SR-BCH codes}).
\end{proposition}
\begin{proof}
First, the linearized Reed-Solomon code $ \mathcal{C}^\sigma_{\delta - 1}(\mathcal{A}, \mathcal{B}) $ is a CSC code by Theorem \ref{th CSC linearized RS codes}. The reader may check that its dual code and the restriction of such a dual code to $ \mathbb{F}^n $ are also CSC codes. Thus $ \mathcal{C}_\delta(a^b,\beta) $ is a CSC code. 

Finally, if $ \mathcal{C} \subseteq \mathbb{F}^n $ is another CSC code whose defining set contains the pairs in (\ref{eq pairs for SR-BCH codes}), then it holds that $ \mathcal{C} \subseteq \mathcal{C}_\delta(a^b,\beta) $ by the proof of Theorem \ref{th SR-BCH bound}, and we are done.
\end{proof}

\subsection{The Defining Set of a SR-BCH Code} \label{subsec SR BCH codes}

In this subsection, we will give a method for finding the defining set of a SR-BCH code directly from the pairs in (\ref{eq pairs for SR-BCH codes}), without explicitly computing its minimal generator skew polynomial (Theorem \ref{th SR-BCH structure}). Thanks to Theorem \ref{th SR-BCH structure}, we will give, in the following subsection, a lower bound on the dimension of SR-BCH codes that is easy to compute from the pairs (\ref{eq pairs for SR-BCH codes}).

We will need the notion of minimal linearized polynomial, which we consider in a slightly more general form than in \cite[Sec. 3.4]{lidl}.

\begin{definition} \label{def minimum linearized polynomial}
Given an extension $ K \subseteq L $ of finite fields of characteristic $ p $ and an arbitrary set $ \mathcal{B} \subseteq L $, we say that $ G(y) \in \mathcal{L}_q K [y] $ (Definition \ref{def linearized pols}) is the minimal linearized polynomial of $ \mathcal{B} $ in $ \mathcal{L}_q K [y] $ if it is the linearized polynomial of minimum degree in $ \mathcal{L}_q K [y] $ such that it is monic and $ G(\beta) = 0 $, for all $ \beta \in \mathcal{B} $.
\end{definition}

The minimal linearized polynomial of a given set exists for any finite-field extension of characteristic $ p $ (to prove this, consider, e.g., $ y^{p^e} - y \in \mathcal{L}_q K[y] $, where $ e = \log_p(q) \log_p|L| $). Its uniqueness follows from the next lemma, whose proof follows the same lines as \cite[Th. 3.68]{lidl} and is left to the reader.

\begin{lemma} \label{lemma minimal lin pol divides on the right}
With notation as in Definition \ref{def minimum linearized polynomial}, let $ F(y), G(y) \in \mathcal{L}_q K[y] $ be such that $ F(\beta) = 0 $, for all $ \beta \in \mathcal{B} $, and $ G(y) $ is the minimal linearized polynomial of $ \mathcal{B} $ in $ \mathcal{L}_q K[y] $. If $ F(y) $ and $ G(y) $ are the associated linearized polynomials of $ f(z),g(z) \in K[z;\sigma] $, respectively (see Definition \ref{def linearized pols}), then $ g(z) $ divides $ f(z) $ on the right in $ K[z;\sigma] $.
\end{lemma}

We will also need the following three auxiliary lemmas. 

\begin{lemma} \label{lemma largest CSC code for given roots}
Let $ x^\ell - 1 = m_1(x) m_2(x) \cdots m_t(x) $ be the irreducible decomposition of $ x^\ell - 1 $ in $ \mathbb{F}[x] $, and choose $ a_i \in \mathbb{F}_q $ such that $ m_i(a_i) = 0 $, for $ i = 1,2, \ldots, t $. For each $ i = 1,2, \ldots, t $, let $ \mathcal{B}_i = \{ \beta_{i,1}, \beta_{i,2}, \ldots, \beta_{i,k_i} \} \subseteq \mathbb{F}_{q^m} \setminus \{ 0 \} $ be a set that does not need to be linearly independent over $ \mathbb{F}_q $. 

Let $ G_i(y) $ be the minimal linearized polynomial of $ \mathcal{B}_i $ in $ \mathcal{L}_q \mathbb{F}_{q_0^{md_i}}[y] $ (Definition \ref{def minimum linearized polynomial}), and assume that it is the associated linearized polynomial of $ g_i(z) \in \mathbb{F}_{q_0^{md_i}}[z;\sigma] $ (Definition \ref{def linearized pols}), for $ i = 1,2, \ldots, t $. Finally, let $ \widetilde{g}_i(x,z) \in \mathcal{S} [z;\sigma] $ be such that its projection onto $ \mathcal{S}_i [z;\sigma] $ is $ g_i(x,z) = {\rm Ev}_{a_i,z}^{-1} (g_i(z)) $, where $ {\rm Ev}_{a_i,z} : \mathcal{S}_i [z;\sigma] \longrightarrow \mathbb{F}_{q_0^{md_i}}[z ; \sigma] $ is the ring isomorphism from (\ref{eq partial eval map as ring isomorphism no quotient}), for $ i = 1,2, \ldots, t $. Then the skew polynomial
\begin{equation}
g(x,z) = \sum_{i=1}^t e_i(x) \widetilde{g}_i(x,z) \in \mathcal{S}[z; \sigma]
\label{eq g(x,z) for BCH}
\end{equation}
is the minimal generator skew polynomial of the largest CSC code $ \mathcal{C} \subseteq \mathcal{R}$, with respect to set inclusion, whose defining set $ T_\mathcal{C} $ contains the pairs
\begin{equation}
\left( a_i, \beta_{i,j} \right) \in \mathbb{F}_q \times ( \mathbb{F}_{q^m} \setminus \{ 0 \} ), 
\label{eq pairs for not necessarily BCH}
\end{equation}
for $ j = 1,2, \ldots, k_i $, for $ i = 1,2, \ldots, t $.
\end{lemma}
\begin{proof}
Fix an index $ i = 1,2, \ldots, t $. By Lemma \ref{lemma minimal lin pol divides on the right} and the paragraph prior to it, $ g_i(z) \in \mathbb{F}_{q_0^{md_i}}[z;\sigma] $ exists, it is unique and furthermore, it divides $ z^m - 1 $ on the right in $ \mathbb{F}_{q_0^{md_i}}[z;\sigma] $, since $ \beta_{i,j}^{q^m} = \beta_{i,j} $, for $ j = 1,2, \ldots, k_i $. Since $ g_i(z) $ divides $ z^m - 1 $ on the right in $ \mathbb{F}_{q_0^{md_i}}[z;\sigma] $, it is the minimal generator skew polynomial of the skew-cyclic code that it generates, $ \mathcal{C}_i = (g_i(z)) \subseteq \mathbb{F}_{q_0^{md_i}}[z;\sigma] / (z^m - 1) $, by item 5 in \cite[Th. 2]{ontheroots}. Therefore, $ g_i(x,z) \in \mathcal{S}_i [z;\sigma] $ is the minimal generator skew polynomial of 
$$ \mathcal{C}^{(i)} = {\rm Ev}_{a_i,z}^{-1} \left( \mathcal{C}_i \right) \subseteq \mathcal{R}_i, $$
since $ {\rm Ev}_{a_i,z} $, given as in (\ref{eq partial eval map as ring isomorphism}), is a ring isomorphism preserving degrees in $ z $.

By Theorem \ref{th single generator}, $ g(x,z) \in \mathcal{S} [z;\sigma] $, given as in (\ref{eq g(x,z) for BCH}), is the minimal generator skew polynomial of the CSC code 
$$ \mathcal{C} = \rho^{-1} \left( \mathcal{C}^{(1)} \times \mathcal{C}^{(2)} \times \cdots \times \mathcal{C}^{(t)} \right) \subseteq \mathcal{R}. $$

It is only left to prove that $ \mathcal{C} $ is the largest CSC code in $ \mathcal{R} $ whose defining set $ T_\mathcal{C} $ contains the pairs in (\ref{eq pairs for not necessarily BCH}). By Proposition \ref{prop partial eval gives isomorphism}, it holds that $ \mathcal{C}_i = \mathcal{C}(a_i) $, for $ i = 1,2, \ldots, t $. Let $ c(x,z) \in \mathcal{R} $ be such that $ c(a_i, 1^{\beta_{i,j}}) = 0 $, for $ j = 1,2, \ldots, k_i $, for $ i = 1,2, \ldots, t $. By Lemma \ref{lemma minimal lin pol divides on the right}, $ g_i(z) $ divides $ c(a_i,z) $ on the right in $ \mathcal{R} $, thus $ c(a_i,z) \in \mathcal{C}(a_i) $, for $ i = 1,2, \ldots, t $. By Corollary \ref{cor total eval product map is isom}, we conclude that $ c(x,z) \in \mathcal{C} $, and we are done.
\end{proof}

\begin{lemma} \label{lemma min degree lin polynomial}
Let $ \mathcal{B} = \{ \beta_1, \beta_2, \ldots, \beta_k \} \subseteq \mathbb{F}_{q^m} $ be an arbitrary set, let $ d \in \mathbb{Z} $ be a divisor of $ s $, and define
\begin{equation*}
\begin{split}
\mathcal{U} & = \bigcup_{j = 1}^k \left\lbrace \beta_j^{q_0^{um d}} \mid u = 0 ,1, \ldots, \frac{s}{d} - 1 \right\rbrace , \textrm{ and} \\
\mathcal{V} & = \left\langle \mathcal{U} \right\rangle _{\mathbb{F}_q} \subseteq \mathbb{F}_{q^m} .
\end{split}
\end{equation*}
Then the conventional polynomial
$$ G(y) = \prod_{\beta \in \mathcal{V}} (y - \beta) $$
is the minimal linearized polynomial of $ \mathcal{B} $ in $ \mathcal{L}_q \mathbb{F}_{q_0^{md}}[y] $.
\end{lemma}
\begin{proof}
First, since $ \mathcal{B} \subseteq \mathbb{F}_{q^m} $, then $ \mathcal{V} \subseteq \mathbb{F}_{q^m} $, hence $ G(y) \in \mathbb{F}_{q^m} [y] $. Since $ \mathcal{V} \subseteq \mathbb{F}_{q^m} $ is an $ \mathbb{F}_q $-linear vector space, we deduce from \cite[Th. 3.52]{lidl} that $ G(y) \in \mathcal{L}_q \mathbb{F}_{q^m}[y] $.

Next, let $ F(y) \in \mathcal{L}_q \mathbb{F}_{q^m}[y] $ be obtained from $ G(y) $ by raising each coefficient of $ G(y) $ to the $ q_0^{md} $th power. Then $ F(\beta^{q_0^{md}}) = G(\beta)^{q_0^{md}} = 0 $, for all $ \beta \in \mathcal{V} $. Since $ \mathcal{U}^{q_0^{md}} = \mathcal{U} $, then $ \mathcal{V}^{q_0^{md}} = \mathcal{V} $, and therefore we deduce that $ F(\beta) = 0 $, for all $ \beta \in \mathcal{V} $. Hence $ F(y) = G(y) $, thus the coefficients of $ G(y) $ lie in $ \mathbb{F}_{q_0^{md}} $. In other words, it holds that $ G(y) \in \mathcal{L}_q \mathbb{F}_{q_0^{md}}[y] $.

Now, let $ E(y) \in \mathcal{L}_q \mathbb{F}_{q_0^{md}}[y] $ be such that $ E(\beta_i) = 0 $, for $ i = 1,2, \ldots, k $. Then 
$$ 0 = E(\beta_i)^{q_0^{umd}} = E \left( \beta_i^{q_0^{umd}} \right), $$
for $ i = 1,2, \ldots, k $, for $ u = 0,1, \ldots, s/d - 1 $. In other words, $ E(\beta) = 0 $, for all $ \beta \in \mathcal{U} $. Since the map $ \beta \mapsto E(\beta) $ is linear over $ \mathbb{F}_q $, we conclude that $ E(\beta) = 0 $, for all $ \beta \in \mathcal{V} $. By the definition of $ G(y) $, it must holds that $ \deg_y(G(y)) \leq \deg_y(E(y)) $, and we are done.
\end{proof}

Next, we relate the defining sets of evaluation skew-cyclic codes on roots of the same irreducible component of $ x^\ell - 1 $. 

\begin{lemma} \label{lemma defining set when changin root unity}
Let $ g(x,z) \in \mathcal{S} [z; \sigma] $ be of degree less than $ m $, and choose $ a, a^\prime \in \mathbb{F}_q $ such that $ m_i(a) = m_i(a^\prime) = 0 $, for some $ i $, $ 1 \leq i \leq t $, where $ x^\ell - 1 = m_1(x) m_2(x) \cdots m_t(x) $ is the irreducible decomposition of $ x^\ell - 1 $ in $ \mathbb{F}[x] $. The following hold:
\begin{enumerate}
\item
There exists $ j = 0, 1, 2, \ldots, d_i - 1 $, where $ d_i = \deg_x(m_i(x)) $, such that $ a^\prime = a^{q_0^{jm}} $.
\item
For $ \beta \in \mathbb{F}_{q^m} \setminus \{ 0 \} $, we have that
$$ g(a, 1^\beta) = 0 \quad \Longleftrightarrow \quad g \left( a^{q_0^{jm}} , 1^{\beta^{q_0^{jm}}} \right) = 0. $$
As a consequence, for a CSC code $ \mathcal{C} \subseteq \mathcal{R} $, we deduce that
$$ T_\mathcal{C} \left( a^{q_0^{jm}} \right) = T_\mathcal{C}(a) ^{q_0^{jm}} , $$
with notation as in (\ref{eq def defining sets of partial ev skew-cyclic codes}). 
\end{enumerate}
\end{lemma}
\begin{proof}
Item 1 is a well-known result \cite[Sec. 4.1]{pless} that follows by observing that $ \prod_{h=0}^{d_i-1} (x - a^{q_0^{hm}}) $ is the polynomial of minimum degree in $ \mathbb{F}[x] $ that has $ a $ as a root, and hence it must coincide with $ m_i(x) $. 

We now prove item 2. By equation (\ref{eq evaluation equation for diff roots}), it holds that $ g(a, 1^\beta) = 0 $ if, and only if, 
\begin{equation}
g_0(a) \beta + g_1(a) \sigma(\beta) + \cdots + g_{m-1}(a) \sigma^{m-1}(\beta) = 0,
\label{eq g to zero proof}
\end{equation}
and similarly in the case of $ a^\prime $, where $ g(x,z) = g_0(x) + g_1(x)z + \cdots + g_{m-1}(x) z^{m-1} $ and $ g_h(x) \in \mathcal{S} $, for $ h = 0,1,\ldots, m-1 $. Since $ g_h(x) \in \mathcal{S} $, it holds that 
$$ g_h \left( a^{q_0^{j m}} \right) = g_h(a)^{q_0^{j m}}, $$
for $ h = 0,1, \ldots , m-1 $. Thus, equation (\ref{eq g to zero proof}) holds if, and only if, we have that
\begin{equation*}
\begin{split}
0 & = \left( g_0(a) \beta + g_1(a) \sigma(\beta) + \cdots + g_{m-1}(a) \sigma^{m-1}(\beta) \right)^{q_0^{j m}} \\
 & = g_0(a)^{q_0^{j m}} \beta^{q_0^{j m}} + \cdots + g_{m-1}(a)^{q_0^{j m}} \sigma^{m-1}(\beta)^{q_0^{j m}} \\
 & = g_0 \left( a^{q_0^{j m}} \right) \beta^{q_0^{j m}} + \cdots + g_{m-1} \left( a^{q_0^{j m}} \right) \sigma^{m-1} \left( \beta^{q_0^{j m}} \right) ,
\end{split}
\end{equation*}
and the result follows.
\end{proof}

We may finally find the defining set of a SR-BCH in terms of the pairs (\ref{eq pairs for SR-BCH codes}), without explicitly computing its minimal generator skew polynomial.

\begin{theorem} \label{th SR-BCH structure}
Let $ a \in \mathbb{F}_q $ be a primitive $ \ell $th root of unity and let $ \beta \in \mathbb{F}_{q^m} $ be a normal element of the extension $ \mathbb{F}_q \subseteq \mathbb{F}_{q^m} $ (see Subsection \ref{subsec finding CSC lin RS codes}). Fix integers $ b \geq 0 $ and $ 2 \leq \delta \leq n $. Let $ x^\ell - 1 = m_1(x) m_2(x) \cdots m_t(x) $ be the irreducible decomposition of $ x^\ell - 1 $ in $ \mathbb{F}[x] $. Define
\begin{equation*}
\begin{split}
\mathcal{J}_i & = \left\lbrace j \in \mathbb{N} \mid 0 \leq j \leq \delta - 2, m_i ( a^{b+j} ) = 0 \right\rbrace \\
 & = \left\lbrace j_1, j_2, \ldots, j_{k_i} \right\rbrace ,
\end{split}
\end{equation*}
where $ k_i = | \mathcal{J}_i | $, and choose an arbitrary $ \widetilde{j}_i \in \mathcal{J}_i $, for $ i = 1,2, \ldots, t $. There exist integers $ h_1, h_2, \ldots, h_{k_i} \in \mathbb{Z} $, satisfying that $ 0 \leq h_\lambda \leq d_i - 1 $ and
\begin{equation}
b + \widetilde{j}_i \equiv (b + j_\lambda) q_0^{h_\lambda m} \quad (\textrm{mod } \ell),
\label{eq congruence for defining set SR-BCH codes}
\end{equation}
for $ \lambda = 1,2, \ldots, k_i $, for $ i = 1,2, \ldots, t $. Define the $ \mathbb{F}_q $-linear vector subspace of $ \mathbb{F}_{q^m} $:
\begin{equation*}
\begin{split}
\mathcal{V}_i = \left\langle \beta^{q_0^v} \mid v \in \{ \right. & sj_\lambda + m (ud_i + h_\lambda) \textrm{ (mod } sm) \mid \\
& \left. u = 0,1, \ldots, \frac{s}{d_i} - 1, \lambda = 1,2, \ldots, k_i \} \right\rangle_{\mathbb{F}_q} ,
\end{split}
\end{equation*}
for $ i = 1,2, \ldots, t $. The following properties hold for the corresponding SR-BCH code $ \mathcal{C}_\delta(a^b,\beta) \subseteq \mathbb{F}^n $ (Definition \ref{def SR-BCH codes}):
\begin{enumerate}
\item
The defining set of $ \mathcal{C}_\delta(a^b,\beta) $ satisfies that
$$ T_{\mathcal{C}_\delta(a^b,\beta)} \left( a^{b + \widetilde{j}_i} \right) = \mathcal{V}_i \setminus \{ 0 \} , $$
with notation as in (\ref{eq def defining sets of partial ev skew-cyclic codes}), for $ i = 1,2, \ldots, t $.
\item
The dimension of $ \mathcal{C}_\delta(a^b,\beta) $ over $ \mathbb{F} $ is
$$ \dim_\mathbb{F} \left( \mathcal{C}_\delta(a^b,\beta) \right) = n - \sum_{i = 1}^t \deg_x(m_i(x)) \dim_{\mathbb{F}_q} (\mathcal{V}_i). $$
\end{enumerate}
\end{theorem}
\begin{proof}
Denote $ \mathcal{C} = \mathcal{C}_\delta(a^b,\beta) $ for simplicity. 

The existence of $ h_1, h_2, \ldots, h_{k_i} \in \mathbb{Z} $ satisfying (\ref{eq congruence for defining set SR-BCH codes}) and $ 0 \leq h_\lambda \leq d_i - 1 $, for $ \lambda = 1,2, \ldots, k_i $, for $ i = 1,2, \ldots, t $, follows from item 1 in Lemma \ref{lemma defining set when changin root unity} and the fact that $ a $ is a primitive $ \ell $th root of unity. By item 2 in Lemma \ref{lemma defining set when changin root unity}, it holds that
\begin{equation*}
\begin{split}
\beta^{q_0^{s j_\lambda}} & \in T_\mathcal{C} \left( a^{b + j_\lambda} \right) \quad \Longleftrightarrow \\
 \beta^{q_0^{s j_\lambda + mh_\lambda}} & \in T_\mathcal{C} \left( \left( a^{b + j_\lambda} \right)^{q_0^{h_\lambda m}} \right) = T_\mathcal{C} \left( a^{b + \widetilde{j}_i} \right).
\end{split}
\end{equation*}
In other words, the fact that $ T_\mathcal{C} $ contains the pairs in (\ref{eq pairs for SR-BCH codes}) is equivalent to the fact that $ T_\mathcal{C} $ contains the pairs
$$ \left( a^{b + \widetilde{j}_i} , \beta^{q_0^{sj_\lambda + mh_\lambda}} \right), $$
for $ \lambda = 1,2, \ldots, k_i $, for $ i = 1,2, \ldots, t $. Hence, by Lemma \ref{lemma largest CSC code for given roots}, since $ \mathcal{C} $ is the largest CSC code containing such pairs (Proposition \ref{prop SR-BCH are largest containing pairs}), we deduce that the linearized polynomial associated to the minimal generator skew polynomial of $ \mathcal{C}(a^{b + \widetilde{j}_i}) $ is the minimal linearized polynomial of 
\begin{equation*}
\begin{split}
\mathcal{B}_i = \lbrace \beta^{q_0^v} \mid v \in \{ & sj_\lambda + mh_\lambda  \textrm{ }(mod \textrm{ } sm) \mid \\
& \lambda = 1,2, \ldots, k_i \} \rbrace \subseteq \mathbb{F}_{q^m}
\end{split}
\end{equation*}
in $ \mathcal{L}_q \mathbb{F}_{q_0^{md_i}}[y] $, for $ i = 1,2, \ldots, t $. By Lemma \ref{lemma min degree lin polynomial}, such a minimal linearized polynomial is 
$$ G_i(y) = \prod_{\widetilde{\beta} \in \mathcal{V}_i} \left( y - \widetilde{\beta} \right) , $$
for $ i = 1,2, \ldots, t $. By the definition of $ T_\mathcal{C}(a^{b + \widetilde{j}_i}) $ (see (\ref{eq def defining sets of partial ev skew-cyclic codes})), Lemma \ref{lemma connection arith eval lin eval} and Proposition \ref{prop partial eval gives isomorphism}, we conclude that item 1 holds, i.e., $ T_\mathcal{C}(a^{b + \widetilde{j}_i}) = \mathcal{V}_i \setminus \{ 0 \} $, for $ i = 1,2, \ldots, t $. 

Finally, item 2 follows from item 1 and Theorem \ref{th dimensions from defining set}.
\end{proof}

\begin{remark}
If, for some $ i = 1,2, \ldots, t $, it holds that $ \mathcal{J}_i = \varnothing $, then $ \mathcal{B}_i = \varnothing $ and $ \mathcal{V}_i = \{ 0 \} $. Hence such a term does not appear in the sum in item 2 in Theorem \ref{th SR-BCH structure}.
\end{remark}

As usual, we conclude by discussing the cases $ m = 1 $ and $ \ell = 1 $. In the classical cyclic case $ m = 1 $, Theorem \ref{th SR-BCH structure} simply says that the defining set of a BCH code is the union of the cyclotomic sets that have a non-empty intersection with the pairs in (\ref{eq pairs for SR-BCH codes}) \cite[Eq. (5.1)]{pless}. In the skew-cyclic case $ \ell = 1 $, rank-metric BCH codes were also defined in terms of defining sets in \cite[Def. 7]{ontheroots}. However, the description in Theorem \ref{th SR-BCH structure} is new in this case to the best of our knowledge. Note that, setting $ s = 1 $ if $ \ell = 1 $, then $ t=1 $ and $ \mathcal{V}_1 = \langle \beta, \sigma(\beta), \ldots, \sigma^{\delta-2}(\beta) \rangle_{\mathbb{F}_q} $, hence Theorem \ref{th SR-BCH structure} is consistent with \cite[Th. 6]{ontheroots} for full-length skew-cyclic Gabidulin codes.

\subsection{A Lower Bound on the Dimension of a SR-BCH Code} \label{subsec estimates}

In this subsection, we will make use of Theorem \ref{th SR-BCH structure} to obtain a simple lower bound on the dimension of SR-BCH codes. This bound only makes use of the first component of the pairs in (\ref{eq pairs for SR-BCH codes}), and can be easily computed by using the corresponding cyclotomic sets. We will show how to use it in Example \ref{ex lower bound}, and we will provide tables for a wide range of parameters in Appendix \ref{app tables}. In those tables, it can be seen how our lower bound provides codes with a higher dimension for a given minimum sum-rank distance than previously known (\ref{eq lower bound delsarte}).

\begin{theorem} \label{th lower bound}
With assumptions and notation as in Theorem \ref{th SR-BCH structure}, it holds that
\begin{equation}
\dim_\mathbb{F} \left( \mathcal{C}_\delta (a^b, \beta) \right) \geq n - \sum_{i = 1}^t d_i \min \left\lbrace m , \frac{s k_i}{d_i} \right\rbrace,
\label{eq lower bound}
\end{equation}
where $ d_i = \deg_x(m_i(x)) $ as in Theorem \ref{th SR-BCH structure}, and
$$ k_i = \left| \left\lbrace j \in \mathbb{N} \mid 0 \leq j \leq \delta - 2, m_i ( a^{b+j} ) = 0 \right\rbrace \right| \geq 0 , $$
for $ i = 1,2, \ldots, t $.
\end{theorem}

We now briefly discuss the bound (\ref{eq lower bound}). First, in the classical cyclic case $ m = 1 $, the bound (\ref{eq lower bound}) is an equality and becomes the well-known formula 
$$ \dim_\mathbb{F} \left( \mathcal{C}_\delta(a^b, 1) \right) = n - \sum_{i=1}^t d_i \varepsilon_i, $$
where $ \varepsilon_i = 1 $ if there exists an integer $ j $ such that $ 0 \leq j \leq \delta - 2 $ and $ m_i(a^{b+j}) = 0 $, and $ \varepsilon_i = 0 $ otherwise, for $ i = 1,2, \ldots, t $. In the skew-cyclic case $ \ell = 1 $, setting $ s = 1 $ (which can always be done), we recover the dimension of full-length ($ N = m $) skew-cyclic Gabidulin codes
$$ \dim_\mathbb{F} \left( \mathcal{C}_\delta (1, \beta) \right) = n - k_1 = n - \delta + 1, $$
since $ t = s = d_1 = 1 $ and $ k_1 = \delta - 1 < n = m $.

As in the classical cyclic case, SR-BCH are in general subfield subcodes of duals of linearized Reed-Solomon codes. Therefore, we may also apply Delsarte's lower bound \cite{delsarte} (see also \cite[Cor. 9]{universal-lrc}) on the dimension of SR-BCH codes, obtaining
\begin{equation}
\dim_\mathbb{F} \left( \mathcal{C}_\delta (a^b, \beta) \right) \geq n - s (\delta - 1).
\label{eq lower bound delsarte}
\end{equation}
However, the bound (\ref{eq lower bound}) is always tighter, since $ \sum_{i=1}^t k_i = \delta - 1 $, thus
\begin{equation}
n - \sum_{i = 1}^t d_i \min \left\lbrace m , \frac{s k_i}{d_i} \right\rbrace \geq n - \sum_{i = 1}^t s k_i = n - s (\delta - 1).
\label{eq comparison lower bound with delsate}
\end{equation}
Observe that equality holds in (\ref{eq comparison lower bound with delsate}) if, and only if, $ s k_i \leq m d_i $, for $ i = 1,2, \ldots ,t $.

Finally, the bound (\ref{eq lower bound}) is always at least $ 0 $ and at most the Singleton bound (\ref{eq singleton bound}) for the prescribed distance $ \delta $:
\begin{equation}
0 \leq n - \sum_{i = 1}^t d_i \min \left\lbrace m , \frac{s k_i}{d_i} \right\rbrace \leq n - \delta + 1.
\label{eq comparison lower bound with singleton}
\end{equation}
The first inequality in (\ref{eq comparison lower bound with singleton}) can be deduced from
$$ \sum_{i = 1}^t d_i \min \left\lbrace m , \frac{s k_i}{d_i} \right\rbrace \leq m \sum_{i=1}^t d_i = m \ell = n, $$
where equality holds if, and only if, $ md_i \leq sk_i $, for $ i = 1,2, \ldots, t $. The second inequality in (\ref{eq comparison lower bound with singleton}) can be deduced from
$$ n - \sum_{i = 1}^t d_i \min \left\lbrace m , \frac{s k_i}{d_i} \right\rbrace \leq n -  \sum_{i = 1}^t k_i = n - \delta + 1, $$
since $ k_i \leq d_i m $, for $ i = 1,2, \ldots, t $. It is left to the reader to show that (\ref{eq lower bound}) may be equal to the Singleton bound (\ref{eq singleton bound}) for the prescribed distance $ \delta $ if, and only if, $ s = 1 $, which is the case in which SR-BCH codes coincide with CSC linearized Reed-Solomon codes as in Theorem \ref{th CSC linearized RS codes}.

We conclude by giving an example of how to compute the bound (\ref{eq lower bound}). This method can easily be automated, and in Appendix \ref{app tables} we provide some tables with values for this bound obtained by a simple implementation in C++.

\begin{example} \label{ex lower bound}
Let $ q_0 = 2 $, $ m = 2 $, thus $ q_0^m = 4 $, $ s = 4 $, $ \ell = 15 $, thus $ n = 30 $. We first compute the $ 4 $-cyclotomic sets modulo $ 15 $, i.e., the sets of powers of $ a $ in the roots of $ m_i(x) = \prod_{h=0}^{d_i-1} (x - a^{q_0^{hm}}) $, for $ i = 1,2, \ldots , t $. Note that we do not need to know $ t $ nor $ m_i(x) $. These cyclotomic sets are $ \{ 0 \} $, $ \{ 1, 4 \} $, $ \{ 2, 8 \} $, $ \{ 3, 12 \} $, $ \{ 5 \} $, $ \{ 6, 9 \} $, $ \{ 7,13 \} $, $ \{ 10 \} $ and $ \{ 11, 14 \} $ (see \cite[Sec. 4.1]{pless}). Choose $ \delta = 5 $ and $ b = 1 $. Then the first components of the pairs in (\ref{eq pairs for SR-BCH codes}) are $ a^1 $, $ a^2 $, $ a^3 $ and $ a^4 $. Thus we may assume that $ k_1 = 2 $ and $ d_1 = 2 $ (corresponding to $ a^1 $ and $ a^4 $), $ k_2 = k_3 = 1 $ and $ d_2 = d_3 = 2 $ (corresponding to $ a^2 $ and $ a^3 $, respectively), being all other $ k_i = 0 $. Thus the lower bound (\ref{eq lower bound}) on the dimension of $ \mathcal{C}_\delta(a^b,\beta) $ yields 
$$ n - d_1 m - s k_2 - s k_3 = 18. $$ 
On the other hand, the Singleton bound (\ref{eq singleton bound}) and Delsarte's bound (\ref{eq lower bound delsarte}) provide $ 26 $ and $ 14 $ as upper and lower bounds on the dimension of the corresponding SR-BCH (see Table \ref{table primitive narrow-sense s=4}). In particular, we beat the previous known lower bound (\ref{eq lower bound delsarte}).
\end{example}

\subsection{Decoding SR-BCH Codes with respect to the Sum-Rank Metric} \label{subsec SR-BCH decoding}

We conclude this manuscript by noting that we may decode SR-BCH codes up to half their prescribed distance by considering them as subfield subcodes of an appropriate linearized Reed-Solomon code.

We start with the following result, whose proof follows the same lines as \cite[Th. 4]{secure-multishot}, and can be of interest by itself.

\begin{proposition} \label{prop duals of CSC lin RS codes}
Let the assumptions and notation be as in Definition \ref{def SR-BCH codes}. In such a case, there exist an integer $ c \geq 0 $ and a normal element $ \gamma \in \mathbb{F}_{q^m} $ of the finite-field extension $ \mathbb{F}_q \subseteq \mathbb{F}_{q^m} $ such that
$$ \mathcal{C}^\sigma_{\delta - 1}(\mathcal{A}, \mathcal{B})^\perp = \mathcal{C}^\sigma_{n - \delta + 1}(\mathcal{A}, \mathcal{B}^\prime), $$
where $ \mathcal{C}^\sigma_{n - \delta + 1}(\mathcal{A}, \mathcal{B}^\prime) \subseteq \mathbb{F}_{q^m}^n $ is as in Definition \ref{def lin RS codes}, for $ \mathcal{A} = \{  1, a, a^2, \ldots,  a^{\ell - 1} \} $ and $ \mathcal{B}^\prime = (\mathcal{B}^\prime_0, \mathcal{B}^\prime_1, \ldots, \mathcal{B}^\prime_{\ell-1}) $ given by $ \mathcal{B}^\prime_i = \{ \gamma a^{ci}, \sigma(\gamma) a^{ci}, \ldots, \sigma^{m-1}(\gamma) a^{ci} \} $, for $ i = 0,1, \ldots, \ell -1 $.
\end{proposition}
\begin{proof}
As in the proof of \cite[Th. 4]{secure-multishot}, there exist $ \alpha^{(i)}_0, $ $ \alpha^{(i)}_1, $ $ \ldots,$ $ \alpha^{(i)}_{m-1} \in \mathbb{F}_{q^m} $ that are linearly independent over $ \mathbb{F}_q $, for $ i = 0 ,1, \ldots , \ell - 1 $, such that
$$ \mathcal{C}^\sigma_{n - 1}(\mathcal{A}, \mathcal{B})^\perp = \left\langle \boldsymbol\alpha \right\rangle_{\mathbb{F}_{q^m}}, $$
where $ \boldsymbol\alpha = (\boldsymbol\alpha^{(0)}, \boldsymbol\alpha^{(1)}, \ldots, \boldsymbol\alpha^{(\ell - 1)}) \in \mathbb{F}_{q^m}^n $, and $ \boldsymbol\alpha^{(i)} = (\alpha^{(i)}_0, \alpha^{(i)}_1, \ldots, \alpha^{(i)}_{m-1}) \in \mathbb{F}_{q^m}^m $, for $ i = 0,1,\ldots, \ell - 1 $. It holds that $ \mathcal{C}^\sigma_{n - 1}(\mathcal{A}, \mathcal{B})^\perp $ is a CSC code, thus there exist elements $ \lambda, \mu \in \mathbb{F}_{q^m} $ such that
$$ \sigma(\alpha^{(i)}_{j-1}) = \lambda \alpha^{(i)}_j \quad \textrm{and} \quad \alpha^{(i-1)}_j = \mu \alpha^{(i)}_j, $$
for $ j = 0,1, \ldots, m-1 $ and $ i = 0,1, \ldots, \ell - 1 $. 

The reader may check that $ \mu^\ell = 1 $, $ \lambda^m = 1 $ and $ \sigma(\lambda) = \lambda $. From $ \mu^\ell = 1 $, we deduce that there exists an integer $ c \geq 0 $ such that $ \mu^{-1} = a^c $ ($ a $ is a primitive $ \ell $th root of unity). From $ \lambda^m = 1 $ and $ \sigma(\lambda) = \lambda $, we deduce that $ \lambda^{\frac{q^m-1}{q-1}} = 1 $. By Hilbert's Theorem 90 \cite[Ex. 2.33]{lidl}, there exists $ \nu \in \mathbb{F}_{q^m} \setminus \{ 0 \} $ such that $ \lambda = \sigma(\nu) / \nu $. Next, define $ \widetilde{\gamma} = \nu^{-1} \alpha^{(0)}_0 \in \mathbb{F}_{q^m} $. It holds that $ \widetilde{\gamma} $ is a normal element of the extension $ \mathbb{F}_q \subseteq \mathbb{F}_{q^m} $, since $ \alpha^{(0)}_0, \alpha^{(0)}_1, \ldots, \alpha^{(0)}_{m-1} \in \mathbb{F}_{q^m} $ are linearly independent over $ \mathbb{F}_q $. We also deduce that
\begin{equation}
\sum_{i = 0}^{\ell - 1} \sum_{j = 0}^{m-1} \sigma^j(\widetilde{\gamma}) a^{ci} a^{li} \sigma^{l+i}(\beta) = 0,
\label{eq proof dual CSC lin RS code}
\end{equation}
for $ l = 0,1, \ldots, n-2 $. Finally, let $ \gamma = \sigma^{-n+\delta}(\widetilde{\gamma}) $, which is also a normal element of the extension $ \mathbb{F}_q \subseteq \mathbb{F}_{q^m} $. Applying $ \sigma^{-u} $ on (\ref{eq proof dual CSC lin RS code}), for $ u = 0,1, \ldots, n-\delta $, we conclude, as in the proof of \cite[Th. 4]{secure-multishot}, that $ \mathcal{C}^\sigma_{\delta - 1}(\mathcal{A}, \mathcal{B})^\perp = \mathcal{C}^\sigma_{n - \delta + 1}(\mathcal{A}, \mathcal{B}^\prime) $, and we are done.
\end{proof}

A general form of dual codes of linearized Reed-Solomon codes was recently given in \cite[Th. 3.2.10]{caruso-dual}. We have decided to include Proposition \ref{prop duals of CSC lin RS codes} since it explicitly treats the particular case of CSC linearized Reed-Solomon codes. Furthermore, we have presented a direct proof that does not need the general theory in \cite{caruso-dual}. 

We conclude that $ \mathcal{C}_\delta (a^b, \beta) = \mathcal{C}_{n-\delta+1}^\sigma(\mathcal{A}, \mathcal{B}^\prime) \cap \mathbb{F}^n $, with assumptions and notation as in Definition \ref{def SR-BCH codes} and Proposition \ref{prop duals of CSC lin RS codes}. Now, we may simply use a decoder for $ \mathcal{C}_{n-\delta+1}^\sigma(\mathcal{A}, \mathcal{B}^\prime) \subseteq \mathbb{F}_{q^m}^n $ with respect to the sum-rank metric for $ \mathbb{F}_q \subseteq \mathbb{F}_{q^m} $, to decode $ \mathcal{C}_\delta (a^b, \beta) \subseteq \mathbb{F}^n $ with respect to the sum-rank metric for $ \mathbb{F}_{q_0} \subseteq \mathbb{F}_{q_0^m} $. 

To see why this approach works, observe that, if 
$$ \mathbf{y} = \mathbf{c} + \mathbf{e} \in \mathbb{F}^n, $$
where $ \mathbf{c} \in \mathcal{C}_\delta (a^b, \beta) $, and $ \mathbf{e} \in \mathbb{F}^n $ is such that $ {\rm wt}^0_{SR}(\mathbf{e}) \leq (\delta - 1)/2 $, then
$$ {\rm wt}_{SR}(\mathbf{e}) \leq {\rm wt}^0_{SR}(\mathbf{e}) \leq \frac{\delta - 1}{2} $$
(see (\ref{eq sum-rank weights in this work}) for the notation). Thus a decoder for $ \mathcal{C}_{n-\delta+1}^\sigma(\mathcal{A}, \mathcal{B}^\prime) $ that corrects up to $ \left\lfloor (\delta - 1)/2 \right\rfloor $ errors with respect to the sum-rank metric $ {\rm wt}_{SR} $ yields $ \mathbf{c} \in \mathcal{C}_{n-\delta+1}^\sigma(\mathcal{A}, \mathcal{B}^\prime) \cap \mathbb{F}^n $, and we are done. 

This approach yields the same complexity as the corresponding decoder of $ \mathcal{C}_{n-\delta+1}^\sigma(\mathcal{A}, \mathcal{B}^\prime) $, also over the larger field $ \mathbb{F}_{q^m} $. Examples of such decoders include \cite{boucher-skew}, \cite[Sec. V]{secure-multishot} and \cite{sven-sum-rank}. They are all extensions of the classical Welch-Berlekamp decoder \cite{welch-berlekamp}, given in decreasing order of computational complexity.

\section{Conclusion and Open Problems} \label{sec conclusion}

In this work, we have introduced the novel families of cyclic-skew-cyclic codes and sum-rank BCH codes. We have studied their structure, obtaining: (1) The minimal generator skew polynomial of a CSC code, with its corresponding generator matrix, (2) the defining set of a CSC code, after carefully considering different evaluation maps, (3) obtained a lower bound (sum-rank BCH bound) on the minimum sum-rank distance of certain CSC codes, and (4) the defining set of sum-rank BCH codes from the pairs in their definition. We have also seen that sum-rank BCH codes can be decoded up to half their prescribed distance by considering them as subfield subcodes of linearized Reed-Solomon codes.

Using their prescribed distance (Theorem \ref{th SR-BCH bound}) and a lower bound on their dimensions based on their defining sets (Theorem \ref{th lower bound}), we obtained in Appendix \ref{app tables} tables of parameters of sum-rank BCH codes, for the finite field $ \mathbb{F}_4 $ and for $ m = 2 $, beating previously known codes.

We now list some open problems for future research:

1) We have made several assumptions on the parameters of CSC codes throughout different sections of the manuscript (see, e.g., the beginning of Section \ref{sec SR BCH codes}). It would be interesting to study CSC codes lifting one or more of these assumptions.

2) The sum-rank BCH bound (Theorem \ref{th SR-BCH bound}) may admit extensions such as the Hartmann-Tzeng bound \cite[Th. 4.5.6]{pless} or the van Lindt-Wilson technique \cite[Th. 4.5.10]{pless}. During the review of this manuscript, this problem was partially solved in \cite{alfarano}, where Hartmann-Tzeng bounds and Roos bounds were given for the sum-rank metric.

3) It would be of interest to find faster decoders of sum-rank BCH codes as in the classical cyclic case, such as the Peterson-Gorenstein-Zierler decoder or the Berlekamp-Massey decoder \cite[Sec. 5.4]{pless}.

4) It would be of interest to find better estimates on the dimension of a sum-rank BCH code than in Theorem \ref{th lower bound}, or better, a simpler formula to exactly compute such a dimension than in Theorem \ref{th SR-BCH structure}.

\section*{Acknowledgement}

The author wishes to thank the anonymous reviewers for their very helpful comments, and especially for the improvement and simplification of Lemma \ref{lemma non-conj over F_q}.

\ifCLASSOPTIONcaptionsoff
  \newpage
\fi



\bibliographystyle{IEEEtran}
%\bibliography{multivariatebib}


\begin{IEEEbiographynophoto}{Umberto Mart{\'i}nez-Pe\~{n}as}
%%Biography text here.
(S’15–M’18) received the B.Sc. and M.Sc. degrees in Mathematics from the University of Valladolid, Spain, in 2013 and 2014, respectively, and he received the Ph.D. degree in Mathematics from Aalborg University, Denmark, in 2017. In 2018 and 2019, he was a Postdoctoral Fellow at the Department of Electrical and Computer Engineering at the University of Toronto, Canada. He is currently a Ma{\^i}tre Assistant at the Institutes of Computer Science and Mathematics at the University of Neuch{\^a}tel, Switzerland. He was awarded an Elite Research Travel Grant (EliteForsk Rejsestipendium) from the Danish Ministery of Education and Science in 2016, the Ph.D. Dissertation Prize by the Danish Academy of Natural Sciences (Danmarks Naturvidenskabelige Akademi, DNA) in 2018, and a Vicent Caselles Prize for Research in Mathematics by the Royal Spanish Mathematical Society (Real Sociedad Matemática Española, RSME) in 2019. His research interests include algebra, algebraic coding, distributed storage, network coding, and information-theoretical security and privacy.
\end{IEEEbiographynophoto}

\appendices

\section{Tables for bound (\ref{eq lower bound}) on some narrow-sense primitive SR-BCH codes} \label{app tables} 

In this appendix, we provide tables of values for the Singleton upper bound (\ref{eq singleton bound}), our lower bound (\ref{eq lower bound}), and Delsarte's lower bound (\ref{eq lower bound delsarte}). All of these are bounds on the dimension $ \dim_\mathbb{F}(\mathcal{C}) $ of a SR-BCH code $ \mathcal{C} \subseteq \mathbb{F}^n $. We only consider the parameters $ m = 2 $, $ q_0 = 2 $, and different values of $ \ell $, $ s $ and $ n $. If we represent codewords as lists of matrices, as in (\ref{eq matrix representation codewords}), then we may see the SR-BCH codes considered in this appendix as subsets
$$ \mathcal{C} \subseteq \left( \mathbb{F}_2^{2 \times 2} \right)^\ell, $$
but being linear over $ \mathbb{F}_4 $ seen as subsets $ \mathcal{C} \subseteq \mathbb{F}_4^n $, $ n = 2 \ell $. We consider $ \ell = q-1 = 2^s - 1 $, extending the concept of \textit{primitive} BCH codes \cite[Sec. 5.1]{pless}, and $ b = 0,1 $, being the latter case $ b=1 $ an extension of \textit{narrow-sense} BCH codes \cite[Sec. 5.1]{pless}. Note also that, since $ \ell = q-1 $ is odd, then $ \ell $ and $ m = 2 $ are coprime, and all assumptions at the beginning of Section \ref{sec SR BCH codes} are satisfied. 

Bold numbers on the column corresponding to Theorem \ref{th lower bound} (\ref{eq lower bound}) mean that the bound (\ref{eq lower bound}) beats (\ref{eq lower bound delsarte}). It is important to notice that the considered Singleton bound (\ref{eq singleton bound}) is with respect the prescribed distance $ \delta $, whereas the exact minimum sum-rank distance $ d $ of the corresponding sum-rank BCH code may be strictly higher. For this reason, it may happen that $ n - d + 1 < n - \delta + 1 $ in some cases. Thus the column corresponding to (\ref{eq singleton bound}) may not represent the tightest Singleton bound for the corresponding dimension.

\begin{table*}
\caption{$ q_0 = 2 $, $ m = 2 $, $ q_0^m = 4 $, $ s = 1 $, $ \ell = 1 $, $ n = 2 $.}
\centering
\begin{tabular}{c|c||c|c|c}
\hline
&&&&\\[-0.8em]
$ \delta $ & $ b $ & Singleton (\ref{eq singleton bound}) & Theorem \ref{th lower bound} (\ref{eq lower bound}) & Delsarte (\ref{eq lower bound delsarte}) \\
\hline\hline
 &&&& \\[-0.8em]
 2 & 0 & 1 & 1 & 1 \\
\hline
 &&&& \\[-0.8em]
 2 & 1 & 1 & 1 & 1 \\
\hline
\end{tabular}
\end{table*}

\begin{table*}
\caption{$ q_0 = 2 $, $ m = 2 $, $ q_0^m = 4 $, $ s = 2 $, $ \ell = 3 $, $ n = 6 $.}
\centering
\begin{tabular}{c|c||c|c|c}
\hline
&&&&\\[-0.8em]
$ \delta $ & $ b $ & Singleton (\ref{eq singleton bound}) & Theorem \ref{th lower bound} (\ref{eq lower bound}) & Delsarte (\ref{eq lower bound delsarte}) \\
\hline\hline
 &&&& \\[-0.8em]
 2 & 0 & 5 & 4 & 4 \\
\hline
 &&&& \\[-0.8em]
 2 & 1 & 5 & 4 & 4 \\
\hline
 &&&& \\[-0.8em]
 3 & 0 & 4 & 2 & 2 \\
\hline
 &&&& \\[-0.8em]
 3 & 1 & 4 & 2 & 2 \\
\hline
\end{tabular}
\end{table*}

\begin{table*}
\caption{$ q_0 = 2 $, $ m = 2 $, $ q_0^m = 4 $, $ s = 3 $, $ \ell = 7 $, $ n = 14 $.}
\centering
\begin{tabular}{c|c||c|c|c}
\hline
&&&&\\[-0.8em]
$ \delta $ & $ b $ & Singleton (\ref{eq singleton bound}) & Theorem \ref{th lower bound} (\ref{eq lower bound}) & Delsarte (\ref{eq lower bound delsarte}) \\
\hline\hline
 &&&& \\[-0.8em]
 2 & 0 & 13 & \textbf{12} & 11 \\
\hline
 &&&& \\[-0.8em]
 2 & 1 & 13 & 11 & 11 \\
\hline
 &&&& \\[-0.8em]
 3 & 0 & 12 & \textbf{9} & 8 \\
\hline
 &&&& \\[-0.8em]
 3 & 1 & 12 & 8 & 8 \\
\hline
 &&&& \\[-0.8em]
 4 & 0 & 11 & \textbf{6} & 5 \\
\hline
 &&&& \\[-0.8em]
 4 & 1 & 11 & 5 & 5 \\
\hline
 &&&& \\[-0.8em]
 5 & 0 & 10 & \textbf{3} & 2 \\
\hline
 &&&& \\[-0.8em]
 5 & 1 & 10 & \textbf{5} & 2 \\
\hline
 &&&& \\[-0.8em]
 6 & 0 & 9 & \textbf{3} & -1 \\
\hline
 &&&& \\[-0.8em]
 6 & 1 & 9 & \textbf{2} & -1 \\
\hline
 &&&& \\[-0.8em]
 7 & 0 & 8 & \textbf{0} & -4 \\
\hline
 &&&& \\[-0.8em]
 7 & 1 & 8 & \textbf{2} & -4 \\
\hline
\end{tabular}
\end{table*}

\begin{table*}
\caption{$ q_0 = 2 $, $ m = 2 $, $ q_0^m = 4 $, $ s = 4 $, $ \ell = 15 $, $ n = 30 $.}
\centering
\begin{tabular}{c|c||c|c|c} 
\hline
&&&&\\[-0.8em]
$ \delta $ & $ b $ & Singleton (\ref{eq singleton bound}) & Theorem \ref{th lower bound} (\ref{eq lower bound}) & Delsarte (\ref{eq lower bound delsarte}) \\
\hline\hline
 &&&& \\[-0.8em]
 2 & 0 & 29 & \textbf{28} & 26 \\
\hline
 &&&& \\[-0.8em]
 2 & 1 & 29 & 26 & 26 \\
\hline
 &&&& \\[-0.8em]
 3 & 0 & 28 & \textbf{24} & 22 \\
\hline
 &&&& \\[-0.8em]
 3 & 1 & 28 & 22 & 22 \\
\hline
 &&&& \\[-0.8em]
 4 & 0 & 27 & \textbf{20} & 18 \\
\hline
 &&&& \\[-0.8em]
 4 & 1 & 27 & 18 & 18 \\
\hline
 &&&& \\[-0.8em]
 5 & 0 & 26 & \textbf{16} & 14 \\
\hline
 &&&& \\[-0.8em]
 5 & 1 & 26 & \textbf{18} & 14 \\
\hline
 &&&& \\[-0.8em]
 6 & 0 & 25 & \textbf{16} & 10 \\
\hline
 &&&& \\[-0.8em]
 6 & 1 & 25 & \textbf{16} & 10 \\
\hline
 &&&& \\[-0.8em]
 7 & 0 & 24 & \textbf{14} & 6 \\
\hline
 &&&& \\[-0.8em]
 7 & 1 & 24 & \textbf{12} & 6 \\
\hline
 &&&& \\[-0.8em]
 8 & 0 & 23 & \textbf{10} & 2 \\
\hline
 &&&& \\[-0.8em]
 8 & 1 & 23 & \textbf{8} & 2 \\
\hline
 &&&& \\[-0.8em]
 9 & 0 & 22 & \textbf{6} & -2 \\
\hline
 &&&& \\[-0.8em]
 9 & 1 & 22 & \textbf{8} & -2 \\
\hline
 &&&& \\[-0.8em]
 10 & 0 & 21 & \textbf{6} & -6 \\
\hline
 &&&& \\[-0.8em]
 10 & 1 & 21 & \textbf{8} & -6 \\
\hline
 &&&& \\[-0.8em]
 11 & 0 & 20 & \textbf{6} & -10 \\
\hline
 &&&& \\[-0.8em]
 11 & 1 & 20 & \textbf{6} & -10 \\
\hline
 &&&& \\[-0.8em]
 12 & 0 & 19 & \textbf{4} & -14 \\
\hline
 &&&& \\[-0.8em]
 12 & 1 & 19 & \textbf{2} & -14 \\
\hline
 &&&& \\[-0.8em]
 14 & 0 & 17 & \textbf{0} & -22 \\
\hline
 &&&& \\[-0.8em]
 14 & 1 & 17 & \textbf{2} & -22 \\
\hline
\end{tabular}
\label{table primitive narrow-sense s=4}
\end{table*}

\begin{table*}
\caption{$ q_0 = 2 $, $ m = 2 $, $ q_0^m = 4 $, $ s = 5 $, $ \ell = 31 $, $ n = 62 $.}
\centering
\begin{tabular}{c|c||c|c|c}
\hline
&&&&\\[-0.8em]
$ \delta $ & $ b $ & Singleton (\ref{eq singleton bound}) & Theorem \ref{th lower bound} (\ref{eq lower bound}) & Delsarte (\ref{eq lower bound delsarte}) \\
\hline\hline
 &&&& \\[-0.8em]
 2 & 0 & 61 & \textbf{60} & 57 \\
\hline
 &&&& \\[-0.8em]
 2 & 1 & 61 & 57 & 57 \\
\hline
 &&&& \\[-0.8em]
 3 & 0 & 60 & \textbf{55} & 52 \\
\hline
 &&&& \\[-0.8em]
 3 & 1 & 60 & 52 & 52 \\
\hline
 &&&& \\[-0.8em]
 4 & 0 & 59 & \textbf{50} & 47 \\
\hline
 &&&& \\[-0.8em]
 4 & 1 & 59 & 47 & 47 \\
\hline
 &&&& \\[-0.8em]
 5 & 0 & 58 & \textbf{45} & 42 \\
\hline
 &&&& \\[-0.8em]
 5 & 1 & 58 & \textbf{47} & 42 \\
\hline
 &&&& \\[-0.8em]
 6 & 0 & 57 & \textbf{45} & 37 \\
\hline
 &&&& \\[-0.8em]
 6 & 1 & 57 & \textbf{42} & 37 \\
\hline
 &&&& \\[-0.8em]
 7 & 0 & 56 & \textbf{40} & 32 \\
\hline
 &&&& \\[-0.8em]
 7 & 1 & 56 & \textbf{37} & 32 \\
\hline
 &&&& \\[-0.8em]
 8 & 0 & 55 & \textbf{35} & 27 \\
\hline
 &&&& \\[-0.8em]
 8 & 1 & 55 & \textbf{32} & 27 \\
\hline
 &&&& \\[-0.8em]
 10 & 0 & 53 & \textbf{30} & 17 \\
\hline
 &&&& \\[-0.8em]
 10 & 1 & 53 & \textbf{27} & 17 \\
\hline
 &&&& \\[-0.8em]
 12 & 0 & 51 & \textbf{25} & 7 \\
\hline
 &&&& \\[-0.8em]
 12 & 1 & 51 & \textbf{22} & 7 \\
\hline
 &&&& \\[-0.8em]
 14 & 0 & 49 & \textbf{20} & -3 \\
\hline
 &&&& \\[-0.8em]
 14 & 1 & 49 & \textbf{17} & -3 \\
\hline
 &&&& \\[-0.8em]
 18 & 0 & 45 & \textbf{5} & -23 \\
\hline
 &&&& \\[-0.8em]
 18 & 1 & 45 & \textbf{7} & -23 \\
\hline
 &&&& \\[-0.8em]
 22 & 0 & 41 & \textbf{5} & -43 \\
\hline
 &&&& \\[-0.8em]
 22 & 1 & 41 & \textbf{7} & -43 \\
\hline
 &&&& \\[-0.8em]
 26 & 0 & 37 & \textbf{0} & -63 \\
\hline
 &&&& \\[-0.8em]
 26 & 1 & 37 & \textbf{2} & -63 \\
\hline
 &&&& \\[-0.8em]
 30 & 0 & 33 & \textbf{0} & -83 \\
\hline
 &&&& \\[-0.8em]
 30 & 1 & 33 & \textbf{2} & -83 \\
\hline
\end{tabular}
\end{table*}

\begin{table*}
\caption{$ q_0 = 2 $, $ m = 2 $, $ q_0^m = 4 $, $ s = 6 $, $ \ell = 63 $, $ n = 126 $.}
\centering
\begin{tabular}{c|c||c|c|c}
\hline
&&&&\\[-0.8em]
$ \delta $ & $ b $ & Singleton (\ref{eq singleton bound}) & Theorem \ref{th lower bound} (\ref{eq lower bound}) & Delsarte (\ref{eq lower bound delsarte}) \\
\hline\hline
 &&&& \\[-0.8em]
 2 & 0 & 125 & \textbf{124} & 120 \\
\hline
 &&&& \\[-0.8em]
 2 & 1 & 125 & 120 & 120 \\
\hline
 &&&& \\[-0.8em]
 3 & 0 & 124 & \textbf{118} & 114 \\
\hline
 &&&& \\[-0.8em]
 3 & 1 & 124 & 114 & 114 \\
\hline
 &&&& \\[-0.8em]
 4 & 0 & 123 & \textbf{112} & 108 \\
\hline
 &&&& \\[-0.8em]
 4 & 1 & 123 & 108 & 108 \\
\hline
 &&&& \\[-0.8em]
 5 & 0 & 122 & \textbf{106} & 102 \\
\hline
 &&&& \\[-0.8em]
 5 & 1 & 122 & \textbf{108} & 102 \\
\hline
 &&&& \\[-0.8em]
 6 & 0 & 121 & \textbf{106} & 96 \\
\hline
 &&&& \\[-0.8em]
 6 & 1 & 121 & \textbf{102} & 96 \\
\hline
 &&&& \\[-0.8em]
 7 & 0 & 120 & \textbf{100} & 90 \\
\hline
 &&&& \\[-0.8em]
 7 & 1 & 120 & \textbf{96} & 90 \\
\hline
 &&&& \\[-0.8em]
 10 & 0 & 117 & \textbf{88} & 72 \\
\hline
 &&&& \\[-0.8em]
 10 & 1 & 117 & \textbf{84} & 72 \\
\hline
 &&&& \\[-0.8em]
 14 & 0 & 113 & \textbf{70} & 48 \\
\hline
 &&&& \\[-0.8em]
 14 & 1 & 113 & \textbf{66} & 48 \\
\hline
 &&&& \\[-0.8em]
 22 & 0 & 105 & \textbf{52} & 0 \\
\hline
 &&&& \\[-0.8em]
 22 & 1 & 105 & \textbf{52} & 0 \\
\hline
 &&&& \\[-0.8em]
 30 & 0 & 97 & \textbf{26} & -48 \\
\hline
 &&&& \\[-0.8em]
 30 & 1 & 97 & \textbf{28} & -48 \\
\hline
 &&&& \\[-0.8em]
 38 & 0 & 89 & \textbf{14} & -96 \\
\hline
 &&&& \\[-0.8em]
 38 & 1 & 89 & \textbf{16} & -96 \\
\hline
 &&&& \\[-0.8em]
 46 & 0 & 81 & \textbf{6} & -144 \\
\hline
 &&&& \\[-0.8em]
 46 & 1 & 81 & \textbf{8} & -144 \\
\hline
 &&&& \\[-0.8em]
 54 & 0 & 73 & \textbf{0} & -192 \\
\hline
 &&&& \\[-0.8em]
 54 & 1 & 73 & \textbf{2} & -192 \\
\hline
 &&&& \\[-0.8em]
 62 & 0 & 65 & \textbf{0} & -240 \\
\hline
 &&&& \\[-0.8em]
 62 & 1 & 65 & \textbf{2} & -240 \\
\hline
\end{tabular}
\end{table*}

\begin{table*}
\caption{$ q_0 = 2 $, $ m = 2 $, $ q_0^m = 4 $, $ s = 7 $, $ \ell = 127 $, $ n = 254 $.}
\centering
\begin{tabular}{c|c||c|c|c}
\hline
&&&&\\[-0.8em]
$ \delta $ & $ b $ & Singleton (\ref{eq singleton bound}) & Theorem \ref{th lower bound} (\ref{eq lower bound}) & Delsarte (\ref{eq lower bound delsarte}) \\
\hline\hline
 &&&& \\[-0.8em]
 2 & 0 & 253 & \textbf{252} & 247 \\
\hline
 &&&& \\[-0.8em]
 2 & 1 & 253 & 247 & 247 \\
\hline
 &&&& \\[-0.8em]
 3 & 0 & 252 & \textbf{245} & 240 \\
\hline
 &&&& \\[-0.8em]
 3 & 1 & 252 & 240 & 240 \\
\hline
 &&&& \\[-0.8em]
 4 & 0 & 251 & \textbf{238} & 233 \\
\hline
 &&&& \\[-0.8em]
 4 & 1 & 251 & 233 & 233 \\
\hline
 &&&& \\[-0.8em]
 5 & 0 & 250 & \textbf{231} & 226 \\
\hline
 &&&& \\[-0.8em]
 5 & 1 & 250 & \textbf{233} & 226 \\
\hline
 &&&& \\[-0.8em]
 6 & 0 & 249 & \textbf{231} & 219 \\
\hline
 &&&& \\[-0.8em]
 6 & 1 & 249 & \textbf{226} & 219 \\
\hline
 &&&& \\[-0.8em]
 7 & 0 & 248 & \textbf{224} & 212 \\
\hline
 &&&& \\[-0.8em]
 7 & 1 & 248 & \textbf{219} & 212 \\
\hline
 &&&& \\[-0.8em]
 10 & 0 & 245 & \textbf{210} & 191 \\
\hline
 &&&& \\[-0.8em]
 10 & 1 & 245 & \textbf{205} & 191 \\
\hline
 &&&& \\[-0.8em]
 14 & 0 & 241 & \textbf{189} & 163 \\
\hline
 &&&& \\[-0.8em]
 14 & 1 & 241 & \textbf{184} & 163 \\
\hline
 &&&& \\[-0.8em]
 22 & 0 & 233 & \textbf{154} & 107 \\
\hline
 &&&& \\[-0.8em]
 22 & 1 & 233 & \textbf{149} & 107 \\
\hline
 &&&& \\[-0.8em]
 30 & 0 & 225 & \textbf{112} & 51 \\
\hline
 &&&& \\[-0.8em]
 30 & 1 & 225 & \textbf{107} & 51 \\
\hline
 &&&& \\[-0.8em]
 38 & 0 & 217 & \textbf{91} & -5 \\
\hline
 &&&& \\[-0.8em]
 38 & 1 & 217 & \textbf{86} & -5 \\
\hline
 &&&& \\[-0.8em]
 46 & 0 & 209 & \textbf{70} & -61 \\
\hline
 &&&& \\[-0.8em]
 46 & 1 & 209 & \textbf{65} & -61 \\
\hline
 &&&& \\[-0.8em]
 54 & 0 & 201 & \textbf{42} & -117 \\
\hline
 &&&& \\[-0.8em]
 54 & 1 & 201 & \textbf{44} & -117 \\
\hline
 &&&& \\[-0.8em]
 62 & 0 & 193 & \textbf{28} & -173 \\
\hline
 &&&& \\[-0.8em]
 62 & 1 & 193 & \textbf{23} & -173 \\
\hline
\end{tabular}
\end{table*}



%\newpage


%\vfill


\end{document}